\documentclass[10pt, journal]{IEEEtran}
\usepackage{amsfonts,amsmath,amssymb}
\usepackage{bm,color,comment,graphicx,epstopdf,mathtools,pdfpages,soul}
\usepackage[bottom]{footmisc}
\usepackage{balance}

\newtheorem{algorithm}{Algorithm}

\newtheorem{definition}{Definition}

\newtheorem{lemma}{Lemma}
\newtheorem{proposition}{Proposition}
\newtheorem{remark}{Remark}
\newtheorem{theorem}{Theorem}

\newenvironment{proof}{\begin{IEEEproof}}{\end{IEEEproof}}
\renewcommand{\IEEEQED}{\IEEEQEDopen}

\newcommand{\QED}{\hfill \IEEEQED}
\newcommand{\Nj}{\mathcal{N}_j}
\newcommand{\tn}{\hat\theta_N}
\newcommand{\vtn}{\hat\vartheta_N}
\newcommand{\zn}{\hat\zeta_N^n}
\newcommand{\zo}{\zeta^{n0}}

\newcommand{\q}{q^{-1}}
\newcommand{\z}{z^{-1}}
\newcommand{\R}{\mathbb{R}}
\newcommand{\E}{\mathbb{E}}
\newcommand{\Eb}{\bar{\E}}

\allowdisplaybreaks

\title{Identification of diffusively coupled linear networks through structured polynomial models}

\author{E.M.M. (Lizan) Kivits and Paul M.J. Van den Hof 
  \thanks{Paper submitted for publication in IEEE Trans. Automatic Control, 7 June 2021. Revised 21 January 2022, 12 May 2022, and 1 July 2022. }
  \thanks{This project has received funding from the European Research Council (ERC), Advanced Research Grant SYSDYNET, under the European Union's Horizon 2020 research and innovation programme (grant agreement No 694504).}
  \thanks{Lizan Kivits and Paul Van den Hof are with the Department of Electrical Engineering, Eindhoven University of Technology, Eindhoven, The Netherlands {\tt\small\{e.m.m.kivits, p.m.j.vandenhof\}@tue.nl}.}
}

\begin{document}

\maketitle
\thispagestyle{empty}
\pagestyle{empty}

\begin{abstract}
Physical dynamic networks most commonly consist of interconnections of physical components that can be described by diffusive couplings. These diffusive couplings imply that the cause-effect relationships in the interconnections are symmetric and therefore physical dynamic networks can be represented by undirected graphs. This paper shows how prediction error identification methods developed for linear time-invariant systems in polynomial form can be configured to consistently identify the parameters and the interconnection structure of diffusively coupled networks. Further, a multi-step least squares convex optimization algorithm is developed to solve the nonconvex optimization problem that results from the identification method.
\end{abstract}

\section{Introduction} \label{sec:intro}

\subsection{Dynamic Networks}
Physical networks can describe many physical processes from different domains, such as mechanical, magnetic, electrical, hydraulic, acoustic, thermal, and chemical processes. Their dynamic behavior is typically described by undirected dynamic interconnections between nodes, where the interconnections represent diffusive couplings \cite{CHENG2017,DORFLER2013,DORFLER2018}. A corresponding representation can be considered to consist of $L$ interconnected node signals $w_j(t)$, $j=1,\ldots,L$, of which the behavior is described according to a second order vector differential equation:
\begin{multline}\label{eq:msd}
  \mathsf{M}_{j0} \ddot w_j(t) + \mathsf{D}_{j0} \dot w_j(t)
  + \sum_{k\in\Nj} \mathsf{D}_{jk}[\dot w_j(t)-\dot w_k(t)] \\
  + \mathsf{K}_{j0} w_j(t) + \sum_{k\in\Nj} \mathsf{K}_{jk}[w_j(t)-w_k(t)] = \underbrace{r_j(t) + v_j(t)}_{u_j(t)},
\end{multline}
with real-valued coefficients $\mathsf{M}_{j0}\geq0$, $\mathsf{D}_{jk}\geq0$, $\mathsf{K}_{jk}\geq0$, $\mathsf{D}_{jj}=0$; $\mathsf{K}_{jj}=0$; $\Nj$ the set of indices of node signals $w_k(t)$, $k\neq j$, with connections to node signals $w_j(t)$; $u_j(t)$ the external signals composed of measured excitation signals $r_j(t)$ and unmeasured disturbances $v_j(t)$; and with $\dot w_j(t)$ and $\ddot w_j(t)$ the first and second order derivative of the node signals $w_j(t)$, respectively. The diffusive type of coupling induces the symmetry constraints $\mathsf{D}_{jk}=\mathsf{D}_{kj}$ and $\mathsf{K}_{jk}=\mathsf{K}_{kj}$ $\forall j,k$.

An obvious physical example of such a network is the mass-spring-damper system shown in Figure \ref{fig:msd}, in which masses $\mathsf{M}_{j0}$ are interconnected through dampers $\mathsf{D}_{jk}$ and springs $\mathsf{K}_{jk}$ with $k\neq0$ and are connected to the ground with dampers $\mathsf{D}_{j0}$ and springs $\mathsf{K}_{j0}$. The positions $w_j(t)$ of the masses $\mathsf{M}_{j0}$ are the signals that are considered to be the node signals\footnote{Note that a system as the one shown in Figure \ref{fig:msd} would require at least a two-dimensional position vector $w_j(t)$, but for notational convenience and without loss of generality, we will restrict our attention to scalar-valued node signals $w_j(t)$.}. The couplings between the masses are diffusive, because springs and dampers are symmetric components.

\begin{figure}[tb]
  \centering
  \includegraphics[page=2,height=3.3cm]{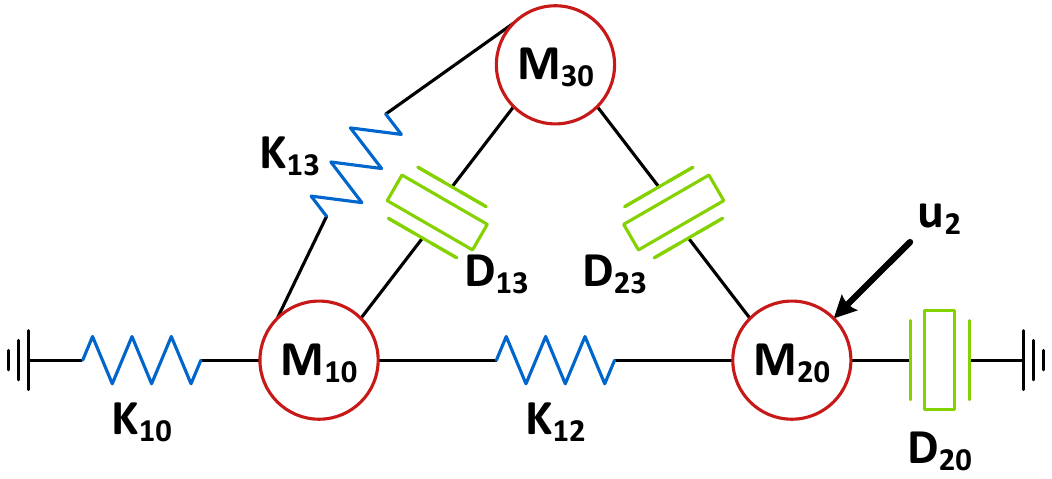}
  \caption{A network of masses ($\mathsf{M}_{j0}$), dampers ($\mathsf{D}_{jk}$) and springs ($\mathsf{K}_{jk}$). } \label{fig:msd}
  \vspace{-1em}
\end{figure}

The considered type of networks occur in many other applications, such as RLC circuits, power grids, and climate control systems.

\subsection{Network identification}
In this paper we address the problem of identifying the physical components in the network on the basis of measured node signals $w(t)$ and possibly external excitation signals $r(t)$. We develop a general framework for identifying such networks. Currently available methods for solving this problem can be classified in different categories. Black-box prediction error identification methods \cite{LJUNG1999} can be used to model the transfer functions from measured $r_j(t)$ signals to node signals $w_j(t)$, leading to a nontrivial second step in which the estimated models need to be converted to the structure representation of the physical network for arriving at estimated component values. Moreover, this modeling procedure and the required conversion would become essentially dependent on the particular location of the external signals $r_j(t)$. In a second category, black-box state-space models can be estimated from which the model parameters can be derived by applying matrix transformations \cite{FRISWELL1999,LOPES2015,RAMOS2013} or eigenvalue decompositions \cite{FRITZEN1986,LUS2003}. However, these methods typically do not have any guarantees on the statistical accuracy of the estimates and lack a consistency analysis. In another approach, state-space models with tailor-made physical parametrizations can be employed in a prediction error / maximum likelihood setting, typically leading to the situation that the network parameters appear nonlinearly in the state-space model, resulting in highly nonconvex optimization problems to solve \cite{Matlab_SYSID}.

In this paper we follow an approach that starts from the network representation of the model, while we maintain and exploit the network structure during the identification procedure.

Dynamic networks are currently topic of research in different areas, while exploiting different network representations. Often, state-space models are used, possibly involving diffusive couplings, e.g. in model reduction \cite{CHENG2021}, estimating network connectivity \cite{TIMME2007}, multi-agent consensus-type algorithms \cite{NABI2012}, and subspace identification \cite{HABER2014}. A different model setting is used by \cite{GONCALVES2007}, where transfer functions are being used to represent the dynamic interactions between node signals, exploited further in \cite{MATERASSI2012,VANWAARDE2021} for topology identification and in \cite{VANDENHOF2013} for prediction error identification of the network dynamics (modules). In the realistic situation that not all states of a network can be measured, the transfer function approach appears attractive for identifying the network, but at the same time it is less fit for representing the physical diffusive type of couplings that would need to be included. As a result, an identification framework that can effectively exploit the physical structure of diffusive couplings, while identifying the dynamics on the basis of selected measurements, is still missing.

\subsection{Research objective and contribution}
The overall objective of this research is to develop a comprehensive theory for the identification of the physical component in diffusively coupled linear networks, where the order of the dynamics is not restricted and possibly correlated disturbances can be present. The objective includes questions like which nodes to measure (sense) and which nodes to excite (actuate) in order to identify particular local dynamics in the network or to identify the full dynamics and topology of the network. In addition, consistency and minimum variance properties of estimates have to be specified. In this paper we focus on the problem of identifying the full dynamics and topology of diffusively coupled linear networks.

We will develop a prediction error framework for identifying the components in a diffusively coupled network of linear time-invariant systems, by fully adhering to the structure of the underlying constitutive model equations. We develop a polynomial representation of diffusively coupled networks, that is special due to its non-monicity and its symmetric structure. For this representation the standard (prediction error) identification algorithms cannot be applied directly. A dedicated prediction error identification method is developed that exploits the structured polynomial representation of the network and allows for handling dynamics of any finite order representing the interconnections and therefore also allows for identifying the topology of the network. New conditions for identifiability and consistent estimation of the network components are derived. While the developed prediction error method in general relies on nonconvex optimization, which is poorly scalable to large dimensions, an alternative multi-step algorithm is presented, following the recent developments in so-called weighted null-space fitting (WNSF) algorithm \cite{GALRINHO2019}. This algorithm is adapted to accommodate the particular structured models that are considered, by involving constrained optimizations rather than unconstrained ones.

This paper builds further on the preliminary work presented in \cite{KIVITS2019}, in which the first results on the polynomial representation are presented in the scope of particular linear regression schemes. These results are extended to the general situation of rational noise models, including detailed identifiability and consistency results as well as the implementation into an adapted WNSF algorithm.

After specifying the notation of our networks in continuous- and discrete-time in Section \ref{sec:phys}, the set-up for identification of the full network dynamics is described in Section \ref{sec:fullidsetup}. In order to be able to consistently identify the network dynamics, data informativity and network identifiability conditions need to be satisfied. These conditions are formulated in Section \ref{sec:fullid} as well as the results for consistent identification of the networks. Section \ref{sec:alg} contains the multi-step algorithm for consistently identifying the network dynamics and Section \ref{sec:simex} consists of a simulation example that illustrates and supports these results. Section \ref{sec:disc} contains some extensions after which conclusions are formulated in Section \ref{sec:conc}.

We consider the following notation throughout the paper. A polynomial matrix $A(\z)$ in complex indeterminate $\z$, consists of matrices $A_{\ell}$ and $(j,k)$-th polynomial elements $a_{jk}(\z)$ such that $A(\z) = \sum_{\ell=0}^{n_a} A_{\ell} z^{-\ell}$ and $a_{jk}(\z) = \sum_{\ell=0}^{n_a} a_{jk,\ell} z^{-\ell}$. Hence, the $(j,k)$-th element of the matrix $A_{\ell}$ is denoted by $a_{jk,\ell}$. Physical components are indicated in sans serif font: $\mathsf{A}$ or $\mathsf{a}$. A $p \times m$ rational function matrix $F(z)$ is proper if $\lim_{z\rightarrow\infty} F(z) = c \in \R^{p\times m}$; it is strictly proper if $c = 0$, and monic if $p=m$ and $c$ is the identity matrix. $F(z)$ is stable if all its poles are within the unit circle $|z| < 1$. As a signal framework we adopt the prediction error framework of \cite{LJUNG1999}, where quasi-stationary signals are defined as summations of a stationary stochastic process and a bounded deterministic signal, and $\Eb := \lim_{N\rightarrow\infty} \frac{1}{N} \sum_{t=1}^{N}\E$, with $\E$ the expectation operator.

\section{Physical network} \label{sec:phys}

\subsection{Higher order network} \label{ssec:high}
A physical network as described in the previous section is typically of second order, where all node signals are collected in $w(t)$. Network models that explain only a selection of the node signals can be constructed by removing nodes from the network through a Gaussian elimination procedure that is referred to as Kron reduction \cite{DORFLER2013,DORFLER2018} or immersion \cite{DANKERS2016}, which will generally lead to higher order dynamics between the remaining node signals. In order to accommodate this, we will include higher order terms in our model.
\begin{definition}[Physical network] \label{def:pn}
  A physical network is a network consisting of $L$ node signals $w_1(t), \ldots, w_L(t)$ interconnected through diffusive couplings and with at least one connection of a node to the ground node. The behavior of the node signals $w_j(t)$, $j=1,\ldots,L$, is described by
  \begin{equation}\label{eq:xy}
    \sum_{\ell=0}^{n_x} \mathsf{x}_{jj,\ell}w_j^{(\ell)}(t) + \sum_{k\in\Nj} \sum_{\ell=0}^{n_y} \mathsf{y}_{jk,\ell}[w_j^{(\ell)}(t)-w_k^{(\ell)}(t)] = u_j(t),
  \end{equation}
  with $n_x$ and $n_y$ the order of the dynamics in the network, with real-valued coefficients $\mathsf{x}_{jj,\ell}\geq0$, $\mathsf{y}_{jk,\ell}\geq0$, $\mathsf{y}_{jk,\ell}=\mathsf{y}_{kj,\ell}$, where $w_j^{(\ell)}(t)$ is the $\ell$-th derivative of $w_j(t)$ and where $u_j(t)$ is the external signal entering the $j$-th node. The network is assumed to be connected, which means that there is a path between every pair of nodes\footnote{The network is connected if its Laplacian matrix (i.e. the degree matrix minus the adjacency matrix) has a positive second smallest eigenvalue \cite{DORFLER2013}.}. \QED
\end{definition}
The graphical interpretation of the coefficients is as follows: $\mathsf{x}_{jj,n}$ represent the buffers, that is the components intrinsically related to the nodes $w_j$; $\mathsf{x}_{jj,\ell}$ with $\ell\neq n$ represent the components connecting the node $w_j$ to the ground node; and $\mathsf{y}_{jk,\ell}$ represent the components in the diffusive couplings between the nodes $w_j$ and $w_k$. The ground node is characterized by $w_{ground}(t)=0$ and therefore can be seen as a node with an infinite buffer, see also \cite{DORFLER2013}.

A graphical representation of a physical network is shown in Figure \ref{fig:pn}. The network dynamics is represented by the blue boxes containing the polynomials $x_{jj}=\sum_{\ell=0}^{n_x}\mathsf{x}_{jj,\ell}p^{\ell}$ and $y_{jk}=\sum_{\ell=0}^{n_y}\mathsf{y}_{jk,\ell}p^{\ell}$, with $p$ the differential operator $d/dt$, and the node signals are represented by the blue circles, which sum the diffusive couplings and the external signals. For example $w_5(t) = \mathsf{x}_{55} \big(w_5(t)-0\big) + \mathsf{y}_{45}\big(w_5(t)-w_4(t)\big) + u_5(t)$.

Furthermore, every matrix $X_{\ell}$ composed of elements $x_{jj,\ell}:=\mathsf{x}_{jj,\ell}$ is diagonal and every matrix $Y_{\ell}$ composed of elements $y_{jj,\ell}:=\sum_{k\in\Nj} \mathsf{y}_{jk,\ell}$ and $y_{jk,\ell}:=-\mathsf{y}_{jk,\ell}$ for $k\neq j$ is Laplacian\footnote{A Laplacian matrix is a symmetric matrix with nonpositive off-diagonal elements and with nonnegative diagonal elements that are equal to the negative sum of all other elements in the same row (or column) \cite{MESBAHI2010}.} representing an undirected graph of a specific physical component, i.e., the diffusive couplings of a specific order.

\begin{figure}
  \centering
  \includegraphics[page=2,height=3.9cm]{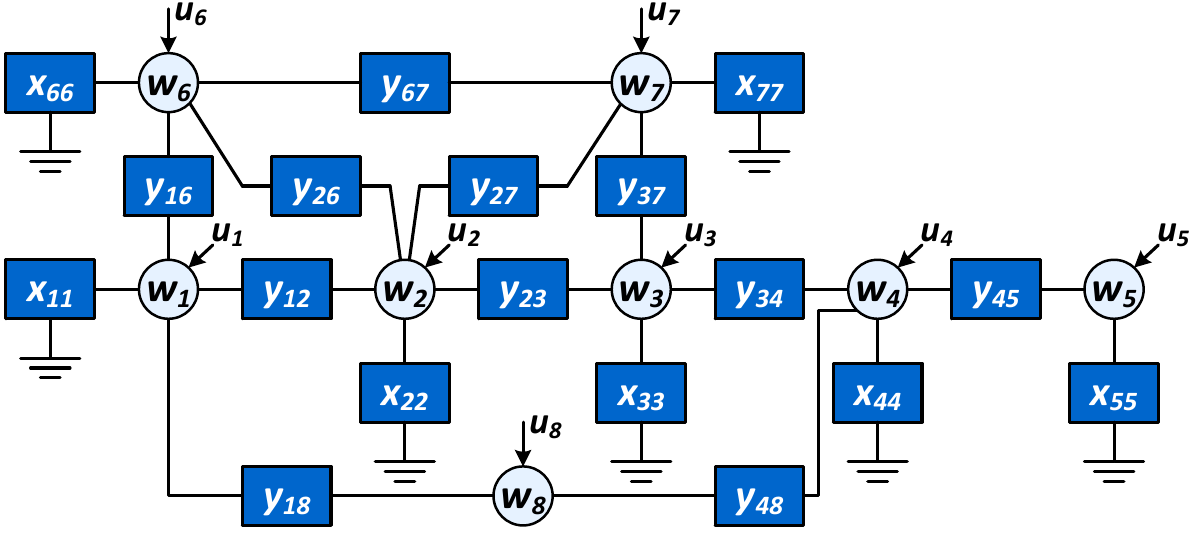}
  \caption{A physical network as defined in Definition \ref{def:pn}, with nodes $w_j$, inputs $u_j$, and dynamics between the nodes ($y_{jk}$) and to the ground node ($x_{jj}$). } \label{fig:pn}
  \vspace{-1em}
\end{figure}

\subsection{Discretization} \label{ssec:dt}
In order to fully exploit the results of network identification that typically have been developed for discrete-time systems, the continuous-time network is converted into an equivalent discrete-time form. Out of the group of discretization methods that commute with series, parallel and feedback connections of systems \cite{MORI1987} we select the backward difference method. This method is relatively simple, results in a causal network representation, and describes a unique bijective mapping between the continuous-time and the discrete-time model, by substituting
\begin{equation} \label{eq:shift}
  \frac{dw(t)}{dt}\biggr\rvert_{t=t_d}=\frac{w(t_d)-w(t_{d-1})}{T_s},
\end{equation}
with discrete-time sequence $t_d = d T_s$, $d=0,1,\ldots$ and time interval $T_s$. Using \eqref{eq:shift}, the continuous-time diffusively coupled network \eqref{eq:xy} can be approximated in discrete time by
\begin{multline}\label{eq:xy_dt}
  \sum_{\ell=0}^{n_x} \bar{\mathsf{x}}_{jj,\ell} q^{-\ell} w_j(t_d) + \sum_{k\in\Nj} \sum_{\ell=0}^{n_y} \bar{\mathsf{y}}_{jk,\ell} q^{-\ell} [w_j(t_d)-w_k(t_d)]\\ = u_j(t_d),
\end{multline}
with $q^{-1}$ the shift operator meaning $q^{-1}w_j(t_d)=w_j(t_{d-1})$ and with
\begin{align}
  \bar{\mathsf{x}}_{jj,\ell}& = (-1)^{\ell} \textstyle\sum_{i=\ell}^{n_x}   \binom{i}{\ell} T_s^{-i} \mathsf{x}_{jj,i},\label{eq:bar1}\\
  \bar{\mathsf{y}}_{jk,\ell}& = (-1)^{\ell} \textstyle\sum_{i=\ell}^{n_y} \binom{i}{\ell} T_s^{-i} \mathsf{y}_{jk,i},\label{eq:bar2}
\end{align}
where $\binom{i}{\ell}$ is a binomial coefficient. In the sequel, $(t-i)$ is used for $t_{d-i}=t_d-iT_s$. The expressions for the node signals \eqref{eq:xy_dt} can be combined in a matrix equation describing the network as
\begin{equation}\label{eq:xy_DT}
  \bar{X}(\q) w(t) + \bar{Y}(\q) w(t) = u(t),
\end{equation}
with $\bar{X}(\q)$ and $\bar{Y}(\q)$ polynomial matrices in the shift operator $q^{-1}$ and composed of elements
\begin{align}
  \bar{X}_{jk}(\q)&=\begin{cases}
    \sum_{\ell=0}^{n_x} \bar{\mathsf{x}}_{jj,\ell} q^{-\ell}, & \mbox{if } k=j \\
    0,                                                      & \mbox{otherwise}
  \end{cases} \label{eq:barX}\\
  \bar{Y}_{jk}(\q)&=\begin{cases}
    \sum_{m\in\Nj} \sum_{\ell=0}^{n_y} \bar{\mathsf{y}}_{jm,\ell} q^{-\ell}, &\mbox{if } k=j\\
    -\sum_{\ell=0}^{n_y} \bar{\mathsf{y}}_{jk,\ell} q^{-\ell},               &\mbox{if } k\in\Nj\\
    0,                                                                           &\mbox{otherwise.}
  \end{cases} \label{eq:barY}
\end{align}
Observe that $\bar{X}(\q)$ is diagonal and $\bar{Y}(\q)$ is Laplacian, implying that the structural properties of \eqref{eq:xy} are maintained in \eqref{eq:xy_DT}-\eqref{eq:barY}.
In the sequel, we will use the notation $A(\q) = \bar{X}(\q)+\bar{Y}(\q)$ while $\bar X(\q)$ and $\bar{Y}(\q)$ can always be uniquely recovered from $A(\q)$, because of their particular structure.

\section{Identification set-up} \label{sec:fullidsetup}
As mentioned before, the objective of this paper is to identify the full dynamics and the topology of diffusively coupled networks. In this section, the identification setting is described, which includes the network model, the network predictor, the model set, and the identification criterion.

The node signals in the network might be affected by a user-applied excitation signal and subject to a disturbance signal. This needs to be included in the network description, which is achieved by splitting the external signal as
\begin{equation}\label{eq:input}
  u(t) := B(\q) r(t) + F(q) e(t),
\end{equation}
where, the known excitation signals $r(t)$ enter the network through dynamics described by polynomial matrix $B(\q)$ and where the unknown disturbance signals acting on the network are modeled as a filtered white noise, i.e. $F(q)$ is a rational matrix and $e(t)$ is a vector-valued wide-sense stationary white noise process, i.e. $\E[e(t)e^T(t-\tau)] = 0$ for $\tau \neq 0$.
\begin{definition}[Network model] \label{def:idsetup}
  The network that will be considered during identification is assumed to be connected, with at least one connection to the ground node; it consists of $L$ node signals $w(t)$ and $K$ excitation signals $r(t)$; and is defined as
  \begin{equation}\label{eq:ABF_DT}
    A(\q) w(t) =  B(\q) r(t) + F(q) e(t),
  \end{equation}
  with
  \begin{itemize}
    \item $A(\q)=\sum_{k=0}^{n_a} A_k q^{-k} \in\mathbb{R}^{L\times L}[\q]$, with $a_{jk}(\q)\! =\! a_{kj}(\q),\forall k,j$ and $A^{-1}(\q)$ stable.
    \item $B(\q)\in \mathbb{R}^{L\times K}[\q]$.
    \item $F(q) \in \mathcal{H}:=\{F\in\mathbb{R}^{L\times L}(q)\ |\ F\ \mbox{monic, stable and}\linebreak \mbox{stably invertible} \}$.
    \item $\Lambda\succ0$ the covariance matrix of the noise $e(t)$.
    \item $r(t)$ is a deterministic and bounded sequence.
    \item $e(t)$ is a zero-mean white noise process with bounded moments of an order larger than 4 (\cite{LJUNG1999}). \QED
  \end{itemize}
\end{definition}

\begin{lemma} \label{lem:rank}
  In \eqref{eq:ABF_DT} it holds that $\text{rank}(A_0)=L$.
\end{lemma}
\begin{proof}
  $A_0 = \bar{X}_0 + \bar{Y}_0$, with diagonal $\bar{X}_0 = \sum_{i=1}^{n_a} T_s^{-i} X_i$ and Laplacian $\bar{Y}_0 = \sum_{i=1}^{n_a} T_s^{-i} Y_i$, following from \eqref{eq:bar1} with $n_x=n_a$ and \eqref{eq:bar2} with $n_y=n_a$, respectively.
    Since $\bar{Y}_0$ is Laplacian, the sum of each row is equal to $0$, that is $\bar{Y}_0 \bm{1} = 0$, with $\bm{1} = \begin{bmatrix}1&1&\ldots&1\end{bmatrix}^{\top} \in \mathbb{R}^L$ \cite{MESBAHI2010}.
  Because the network is connected, $\bar{Y}_0$ has one-dimensional kernel $\ker(\bar{Y}_0) = \text{span}(\bm{1})$ \cite{DORFLER2013}.
    Since the network has at least one connection to the ground node, $\exists j,\ell$ such that $x_{jj,\ell}>0$, implying that $\bar{x}_{jj,0}>0$ and thus $\bar{X}_0\succeq0$.
  The vectors that span the kernel of $\bar{X}_0$ will have at least one zero element, implying that $\ker(\bar{Y}_0)\not\subset\ker(\bar{X}_0)$.
    Because both $\bar{Y}_0\succeq0$ and $\bar{X}_0\succeq0$, $\ker(A_0) = \ker(\bar{Y}_0+\bar{X}_0) = \ker(\bar{Y}_0) \cap \ker(\bar{X}_0) = \emptyset$ and hence, $\text{rank}(A_0)=L$.
\end{proof}

$\text{rank}(A_0)=L$ also implies that $A^{-1}(\q)$ exists and is proper, which means that the network is well-posed. The network is also stable as $A^{-1}(\q)$ is stable.

Often, $B(\q)$ is chosen to be binary, diagonal and known, which represents the assumption that each external excitation signal directly enters the network at a distinct node.

As a result, the considered networks lead to polynomial models\footnote{Polynomial models are linear time-invariant dynamic models of the form $A(\q)y(t) = E^{-1}(\q)B(\q)u(t) + D^{-1}(\q)C(\q)e(t)$, where $A(\q)$, $B(\q)$, $C(\q)$, $D(\q)$, and $E(\q)$ are polynomials in $\q$ that are all monic except for $B(\q)$ \cite{LJUNG1999,HANNAN2012}.} with the particular properties that $A(\q)$ is symmetric and nonmonic. Moreover, if $F(q)$ is polynomial or even stronger if $F(q)=I$, the network \eqref{eq:ABF_DT} leads to an ARMAX-like or ARX-like\footnote{The structure is formally only an ARMAX (autoregressive-moving average with exogenous variables) or ARX (autoregressive with exogenous variables) structure if the $A(\q)$ polynomial is monic \cite{HANNAN2012}. } model structure, respectively.

Now the network representation and its properties have been defined, the next step is to formulate the identification setting.

\subsection{Network predictor}
The objective is to identify the dynamics of the complete network. This estimation is performed using a prediction error method, which is the most common system identification method and it is applicable to networks \cite{VANDENHOF2013}. In order to identify the  network dynamics, all node signals $w(t)$ are predicted based on the measured signals that are available in the network. This leads to the following predictor.
\begin{definition}[Network predictor]
  In line with \cite{WEERTS2016}, the network predictor is defined as the conditional expectation
\begin{equation}\label{eq:npdef}
  \hat{w}(t|t-1) = \E\{w(t) ~|~ w^{t-1}, ~r^t \},
\end{equation}
where $w^{t-1}$ represents the past of $w(t)$, that is $w(t-1), \linebreak w(t-2), \ldots$ and $r^t$ represents $r(t), r(t-1), \ldots$ \QED
\end{definition}
\begin{proposition}[Network predictor]
  For a network model \eqref{eq:ABF_DT}, the one-step-ahead network predictor \eqref{eq:npdef} is given by (omitting arguments $q$)
  \begin{equation}\label{eq:np}
    \hat{w}(t|t-1) :=  \left[I - A_0^{-1} F^{-1}A \right] w(t) + A_0^{-1} F^{-1} B r(t).
  \end{equation}
\end{proposition}
\begin{proof}
  The network \eqref{eq:ABF_DT} can be described by
  \begin{equation*}
    A w(t) = B r(t) + F e(t) = B r(t) + F A_0 A_0^{-1} e(t).
  \end{equation*}
  Premultiplying with $A_0^{-1} F^{-1}$ gives
  \begin{equation*}
    A_0^{-1} F^{-1} A w(t) = A_0^{-1} F^{-1} B r(t) + A_0^{-1} e(t).
  \end{equation*}
  Adding $w(t)$ to both sides of the equality and rewriting gives
  \begin{equation} \label{eq:np1}
     w(t) = \left[I\! -\! A_0^{-1} F^{-1} A \right]\! w(t)\! +\! A_0^{-1} F^{-1} B r(t)\! +\! A_0^{-1} e(t)
  \end{equation}
  where the factor $A_0^{-1}$ makes the filter $\left[I - A_0^{-1} F^{-1} A \right]$ strictly proper and where $A_0^{-1} F^{-1} B$ is proper. The one-step-ahead network predictor \eqref{eq:np} follows directly by applying its definition \eqref{eq:npdef} to \eqref{eq:np1}.
\end{proof}
\begin{proposition}[Innovation] \label{prop:in}
  The innovation corresponding to the network predictor \eqref{eq:np} is
  \begin{equation}\label{eq:in}
    \bar{e}(t) := w(t) - \hat{w}(t|t-1) = A_0^{-1} e(t),
  \end{equation}
  which has covariance matrix $\bar{\Lambda}=A_0 ^{-1}\Lambda A_0 ^{-1}$.
\end{proposition}
\begin{proof}
  This follows directly from subsequently substituting $\hat{w}(t|t-1)$ \eqref{eq:np} and $w(t)$ \eqref{eq:ABF_DT} into \eqref{eq:in}.
\end{proof}
The innovation is a scaled version of the driving noise process. As $A_0 $ is not necessarily diagonal, the scaling possibly causes correlations among the noise channels, but the innovation signal $\bar{e}(t)$ remains a white noise process.

\subsection{Model set and prediction error}
The network models that will be considered during identification are gathered in the network model set.
\begin{definition}[Network model set]
  The model set is defined as a set of parametrized functions as
  \begin{equation}\label{eq:modelset}
    \mathcal{M}:=\{M(\theta),\ \theta\in\Theta\ \subset \mathbb{R}^d \},
  \end{equation}
  with $d \in \mathbb{N}$, with all particular models
  \begin{equation} \label{eq:model}
      M(\theta):=\big(A(\q,\theta),\ B(\q,\theta),\ F(q,\theta), \Lambda(\theta)\big)
  \end{equation}
  satisfying the properties in Definition \ref{def:idsetup}. \QED
\end{definition}
In this setting, $\theta$ contains all the unknown coefficients that appear in the entries of the model matrices $A, B, F$ and $\Lambda$.

The experimental data that are available for identification are generated by the true system.
\begin{definition}[Data generating system] \label{def:dgn}
  The data generating system $\mathcal{S}$ is denoted by the model
  \begin{equation}\label{eq:dgn}
    M^0 := (A^0,B^0,F^0,\Lambda^0).
  \end{equation} \QED
\end{definition}
The true system is in the model set $(\mathcal{S}\in\mathcal{M})$ if $\exists \theta^0\in\Theta$ such that $M(\theta^0)=M^0$, where $\theta^0$ indicate the true parameters.

Using the parametrized network model set, the parametrized one-step-ahead network predictor is defined.
\begin{definition}[Parametrized predictor] \label{def:pp}
  The parametrized network predictor is defined in accordance with \eqref{eq:np} as
  \begin{equation}\label{eq:pnp}
    \hat{w}(t|t-1;\theta) = W(q,\theta) z(t),
  \end{equation}
  with data $z(t) := \begin{bmatrix}w(t)\\r(t)\end{bmatrix}$, and predictor filter
  \begin{equation}\label{eq:pf}
    W(q,\theta) := \begin{bmatrix}I-W_w(q,\theta)& W_r(q,\theta)\end{bmatrix},
  \end{equation}
where
  \begin{align}
    W_w(q,\theta) &= A_0^{-1}(\theta)F^{-1}(q,\theta)A(\q,\theta),    \label{eq:pnpWw}\\
    W_r(q,\theta) &= A_0^{-1}(\theta)F^{-1}(q,\theta)B(\q,\theta).    \label{eq:pnpWr}
  \end{align} \QED
\end{definition}
The parametrized predictor leads to the prediction error.
\begin{proposition}[Prediction error] \label{prop:pe}
  The prediction error corresponding to the parametrized predictor \eqref{eq:pnp} is defined as
  \begin{equation}\label{eq:ped}
    \bar{\varepsilon}(t,\theta) := w(t) - \hat{w}(t|t-1;\theta),
  \end{equation}
  which is obtained as (omitting argument $q$)
  \begin{align}
    \bar{\varepsilon}(t,\theta)
    &= A_0^{-1}(\theta) F^{-1}(\theta) \left[ A(\theta) w(t) -  B(\theta) r(t) \right], \label{eq:pe}\\
    &= W_w(q,\theta) w(t) - W_r(q,\theta) r(t), \label{eq:peW}
  \end{align}
  which equals the innovation $\bar{e}(t)$ \eqref{eq:in} for $\theta=\theta^0$.
\end{proposition}
\begin{proof}
   The expression for the parametrized prediction error \eqref{eq:pe} directly follows from its definition \eqref{eq:ped} and the network predictor \eqref{eq:pnp}. Expressing the parametrized prediction error \eqref{eq:pe} in terms of $r(t)$ and $e(t)$ yields
  \begin{equation}\label{eq:pere}
    \bar{\varepsilon}(t,\theta) = W_{\bar{\varepsilon}r}(q,\theta)r(t) + W_{\bar{\varepsilon}e}(q,\theta)e(t) + (A_0^0)^{-1}e(t),
  \end{equation}
  with (omitting argument $q$)
  \begin{align*}
    W_{\bar{\varepsilon}r}(q,\theta) &= A_0^{-1}(\theta) F^{-1}(\theta) \left[ A(\theta) (A^0)^{-1} B^0 - B(\theta) \right], \\
    W_{\bar{\varepsilon}e}(q,\theta) &= A_0^{-1}(\theta) F^{-1}(\theta) A(\theta) (A^0)^{-1} F^0 - (A_0^0)^{-1}.
  \end{align*}
  The latter two terms in \eqref{eq:pere} are uncorrelated since $e(t)$ is white noise and $W_{\bar{\varepsilon}e}(q,\theta)$ is strictly proper. If the true system is in the model set, the prediction error for the true system is equal to the innovation \eqref{eq:in}:
  \begin{equation}\label{eq:pein}
    \bar{\varepsilon}(t,\theta^0) = (A_0^0)^{-1} e(t) = \bar{e}(t).
  \end{equation}
\end{proof}

\subsection{Identification criterion}\label{ssec:idcrit}
In order to estimate the parameters, a weighted least squares identification criterion is applied:
\begin{align}
  \tn =& \arg \min_{\theta} V_N(\theta), \label{eq:crit} \\
  V_N(\theta) :=& \frac{1}{N} \sum_{t=1}^{N} \bar{\varepsilon}^{\top}(t,\theta) S \bar{\varepsilon}(t,\theta), \label{eq:cost}\\
  \bar{\Lambda}(\tn) :=& \frac{1}{N} \sum_{t=1}^{N} \bar{\varepsilon}(t,\tn) \bar{\varepsilon}^{\top}(t,\tn), \label{eq:cov}
\end{align}
with weight $S\succ0$ that has to be chosen by the user. It is a standard result in prediction error identification (Theorem 8.2 in \cite{LJUNG1999}), that under uniform stability conditions on the parametrized model set, \eqref{eq:crit} converges with probability $1$ to
\begin{align}
\theta^{\ast} := & \arg \min_{\theta} \bar{V}(\theta), \label{eq:critconv} \\
  \mbox{with}\ \ \bar{V}(\theta) :=& \Eb \left\{\bar{\varepsilon}^{\top}(t,\theta) S \bar{\varepsilon}(t,\theta)\right\}. \label{eq:costconv}
\end{align}

\section{Consistent identification} \label{sec:fullid}
In order to consistently identify the network, the experimental data need to satisfy certain conditions. These conditions are referred to as data informativity conditions. In addition, the network itself needs to satisfy certain conditions, such that it can be uniquely recovered. These conditions are referred to as network identifiability conditions. This section describes these conditions, after which the results for consistent network identification can be formulated.

The network \eqref{eq:ABF_DT} can be represented as
\begin{equation}\label{eq:TF}
  w(t) = T_{wr}(q) r(t) + \bar{v}(t),\quad \bar{v}(t) = T_{w\bar{e}}(q)\bar{e}(t),
\end{equation}
where $\bar{e}(t)$ is the innovation \eqref{eq:in} and
\begin{align}
  T_{wr}(q) &= A^{-1}(\q)B(\q), \label{eq:Twr}\\
  T_{w\bar{e}}(q) &= A^{-1}(\q) F(q) A_0 . \label{eq:Twe}
\end{align}
For estimating a network model, prediction error identification methods typically use the second order statistical properties of the measured data, which are represented by the spectral densities of $w(t)$ and $r(t)$. As $r(t)$ is measured, but $\bar{e}(t)$ is not, the second order properties of $w(t)$ are generated by transfer function $T_{wr}(q)$ and spectral density
\begin{align}
  \Phi_{\bar{v}}(\omega):&=\mathcal{F}\{\E[\bar{v}(t)\bar{v}^{\top}(t-\tau)]\}, \label{eq:Phiv}\\
  &=T_{w\bar{e}}(e^{i\omega}) \bar{\Lambda} T_{w\bar{e}}^{\ast}(e^{i\omega}), \label{eq:Phiv2}
\end{align}
with $\mathcal{F}$ the discrete-time Fourier transform and $(\cdot)^{\ast}$ the complex conjugate transpose. Observe that the spectral factorization in \eqref{eq:Phiv2} is unique, as $T_{w\bar{e}}(q)\in\mathcal{H}$ and $\bar{\Lambda}\succ0$ \cite{YOULA1961}.

\subsection{Data informativity}
The data are called informative if they contain sufficient information to uniquely recover the predictor filter $W(q,\theta)$ from second order statistical properties of the data $z(t)$. This can be formalized in line with \cite{LJUNG1999} as follows.
\begin{definition}[Data informativity]\label{def:inf}
 A quasi-stationary data sequence $\{z(t)\}$ is called informative with respect to the model set $\mathcal{M}$ \eqref{eq:modelset} if for any two $\theta_1,\theta_2\in\Theta$
   \begin{multline}\label{eq:inf}
    \Eb\left\{\| [W(q,\theta_1)-W(q,\theta_2)]z(t)\|^2 \right\}=0 \\\Rightarrow \{W(e^{i\omega},\theta_1) = W(e^{i\omega},\theta_2)\}
  \end{multline}
  for almost all $\omega$. \QED
\end{definition}
Applying this definition to physical networks, leads to the following conditions for data informativity.
\begin{proposition}[Data informativity]\label{prop:inf}
  The quasi-stationary data sequence $z(t)$ is informative with respect to the model set $\mathcal{M}$ if $\Phi_z(\omega)\succ0$ for a sufficiently high number of frequencies\footnote{The number of frequencies for which $\Phi_z(\omega) \succ 0$ is required, is dependent on the number of parameters in $\mathcal{M}$.}. In the situation $K\geq1$, this is guaranteed by $\Phi_r(\omega)\succ0$ for a sufficiently high number of frequencies.
\end{proposition}
\begin{proof}
  The premise of implication \eqref{eq:inf} is satisfied if and only if
  $\Delta_W(q,\theta):=W(q,\theta_1)-W(q,\theta_2)=0$. Applying Parseval's theorem gives
  \begin{equation*}
    \frac{1}{2\pi} \int_{-\pi}^{\pi} \Delta_{W}(e^{i\omega},\theta) \Phi_z(\omega) \Delta_{W}^{\top}(e^{-i\omega},\theta) d\omega = 0.
  \end{equation*}
  This implies $\Delta_{W}(q,\theta)=0$ only if $\Phi_z(\omega)\succ0$ for a sufficiently high number of frequencies. In the situation $K\geq1$, $w(t)$ depends on $r(t)$ and substituting the open-loop response \eqref{eq:ABF_DT} for $w(t)$ gives
  \begin{equation*}
    z(t) = J(q) \kappa(t).
  \end{equation*}
  with
 \begin{equation*}
    J(q) = \begin{bmatrix}A^{-1}F& A^{-1}B\\0&I\end{bmatrix}, \qquad
    \kappa(t) = \begin{bmatrix}e(t)\\r(t)\end{bmatrix}.
  \end{equation*}
  As $J(q)$ has always full rank, $\Phi_z(\omega)\succ0$ if and only if $\Phi_{\kappa}(\omega)\succ0$. As $e(t)$ and $r(t)$ are assumed to be uncorrelated and $E\{e(t)\}=0$, we have that $\Phi_{re}=\Phi_{er}=0$ and
  \begin{equation*}
    \Phi_{\kappa} = \begin{bmatrix}\Phi_r&\Phi_{re}\\ \Phi_{er}&\Phi_e\end{bmatrix} = \begin{bmatrix}\Phi_r&0\\0& \Lambda\end{bmatrix}.
  \end{equation*}
  Then $\Phi_{\kappa}(\omega)\succ0$ if and only if $\Lambda\succ0$ (which is assumed) and $\Phi_r(\omega)\succ0$. The condition $\Phi_z(\omega)\succ0$ reduces to $\Phi_r(\omega)\succ0$.
\end{proof}
The condition that $\Phi_r(\omega)\succ0$ for a sufficiently high number of frequencies seems to be a general condition. However, observe that the dimensions of $\Phi_r(\omega)$ depend on the number of excitation signals $r(t)$, denoted by $K$, which is specified in the model set. Thus all excitation signals $r(t)$ that are present (according to the model set), need to be persistently exciting. This is because each additional excitation signal $r_j(t)$ also introduces new polynomials $b_{kj}(\q)$ that need to be identified.

Informativity of $z(t)$ implies that $W(q)$ can uniquely be recovered from data, which by \eqref{eq:pnpWw}-\eqref{eq:pnpWr} and \eqref{eq:TF}-\eqref{eq:Twe} implies that the pair $(T_{wr}(q),\Phi_{\bar v}(\omega))$ can uniquely be recovered from data.

\subsection{Network identifiability}
The concept of network identifiability has been defined for general linear dynamic networks in \cite{WEERTS2018_A1} as follows:
\begin{definition}[Network identifiability]\label{def:idf}
  The network model set $\mathcal{M}$ \eqref{eq:modelset} is globally network identifiable from $z(t)$ if the parametrized model $M(\theta)$ can uniquely be recovered from $T_{wr}(q,\theta)$ and $\Phi_{\bar{v}}(\omega,\theta)$, that is if for all models $M(\theta_1), M(\theta_2) \in \mathcal{M}$
  \begin{equation}\label{eq:idf}\begin{rcases*}
    T_{wr}(q,\theta_1) = T_{wr}(q,\theta_2) \\
    \Phi_{\bar{v}}(\omega,\theta_1) = \Phi_{\bar{v}}(\omega,\theta_2)
    \end{rcases*} \Rightarrow M(\theta_1) = M(\theta_2).
    \end{equation} \QED
\end{definition}
Whereas the original definition has been applied to network models with all transfer function elements, here we apply it to our choice of models \eqref{eq:model}, where through the particular parametrization of the polynomials $A(q,\theta)$ and $B(q,\theta)$, equality of models implies equality of the physical parameters in these polynomial matrices. Before formulating the conditions for identifiability of our particular networks, a result on left matrix fraction descriptions is presented.
\begin{lemma}[Left matrix fraction description (LMFD)]\label{lem:LMFD}
  Consider a network as defined in Definition \ref{def:idsetup}. The LMFD $A(\q)^{-1}B(\q)$ is unique up to a scalar factor if the following conditions are satisfied:
  \begin{enumerate}
    \item The polynomials $A(\q)$ and $B(\q)$ are left coprime.
    \item At least one of the matrices in the set $\{A_k, k = 0, \cdots n_a; B_{\ell}, \ell = 0, \cdots n_b \}$ is diagonal and full rank.
  \end{enumerate}
\end{lemma}
\begin{proof}
  According to \cite{KAILATH1980}, the LMFD of any two polynomial and left coprime matrices is unique up to a premultiplication with a unimodular matrix. To preserve diagonality of $A_k$ or $B_{\ell}$, the unimodular matrix is restricted to be diagonal. To preserve symmetry of $A(\q)$, this diagonal matrix is further restricted to have equal elements.
\end{proof}
In general polynomial models, like ARMAX \cite{DEISTLER1983}, $A(\q)$ is monic and therefore $A_0=I$ is diagonal. Then the LMFD $A(\q)^{-1}B(\q)$ is unique, as the conditions of Lemma \ref{lem:LMFD} are satisfied and scaling with a scalar factor is not possible anymore, since the diagonal elements of $A_0$ are equal to 1. Hence, both Condition 2 in Lemma \ref{lem:LMFD} and the scaling factor freedom are a result of the fact that $A(\q)$ is not necessarily monic.

Now the conditions for global network identifiability can be formulated.
\begin{proposition}[Network identifiability]\label{prop:idf}
  A network model set $\mathcal{M}$ \eqref{eq:modelset} is globally network identifiable from $z(t)$ if the following conditions are satisfied:
  \begin{enumerate}
    \item The polynomials $A(\q)$ and $B(\q)$ are left coprime.
    \item At least one of the matrices in the set $\{A_k, k = 0, \cdots n_a; B_{\ell}, \ell = 0, \cdots n_b  \}$ is diagonal and full rank.
    \item At least one excitation signal $r_j(t)$ is present: $K\geq1$.
    \item There is at least one constraint on the parameters of $A(\q,\theta_a)$ and $B(\q,\theta_a)$ of the form $\tilde{\Gamma}\tilde{\theta}=\gamma\neq0$, with $\tilde{\theta} := \begin{bmatrix}\theta_a^{\top}&\theta_b^{\top}\end{bmatrix}^{\top}$.
  \end{enumerate}
\end{proposition}
\begin{proof}
  Condition 3 implies that $T_{wr}(q,\theta)$ is nonzero. According to Lemma \ref{lem:LMFD}, Condition 1 and 2 imply that $A(\q,\theta)$ and $B(\q,\theta)$ are found up to a scalar factor $\alpha$. $T_{w\bar{e}}(q,\theta)$ and $\bar{\Lambda}(\theta)$ are uniquely recovered from $\Phi_{\bar{v}}(\omega,\theta)$ as $T_{w\bar{e}}(q)\in\mathcal{H}$ and $\bar{\Lambda}\succ0$ \cite{YOULA1961}. Together with the fact that $A(\q,\theta)$ is found up to a scalar factor $\alpha$, $T_{w\bar{e}}(q,\theta)$ gives a unique $F(q,\theta)$, and $\bar{\Lambda}(\theta)$ gives $\Lambda(\theta)$ up to a scalar factor $\alpha^2$.
  Finally, Condition~4 implies that the parameters cannot be scaled anymore and therefore $\alpha$ is fixed.
\end{proof}
The coprimeness of $A(\q)$ and $B(\q)$ ensures that $A(\q)$ and $B(\q)$ have no common factors. This condition is also necessary for global identifiability of typical polynomial model structures, see Theorem 4.1 of \cite{LJUNG1999}. The parameter $\alpha$ is a scaling factor that is introduced by the nonmonicity of $A(\q)$. The scaling factor needs to be fixed by additional constraints induced by Condition 2 and Condition 4 in Proposition \ref{prop:idf}. The parameter constraint in Condition 4 of Proposition \ref{prop:idf} can for example be
\begin{itemize}
  \item One nonzero element in $B(\q,\theta)$ is known, i.e. one excitation signal enters a node through known dynamics.
  \item One nonzero parameter is known.
  \item The fraction of two nonzero parameters is known.
  \item The sum of some nonzero parameters is known.
\end{itemize}

\begin{remark}
In general dynamic networks conditions for global network identifiability typically include algebraic conditions verifying the rank of particular transfer functions from external signals to internal node signals \cite{WEERTS2018_A1}. For the generic version of network identifiability this entails a related graph-based check on vertex disjoint paths in the network model \cite{HENDRICKX2019,CHENG2022}. In contrast to these conditions, the current conditions in Proposition \ref{prop:idf} are very simple and require only a single excitation signal $r(t)$ to be present in the network. This is induced by the structural properties of the diffusive couplings between the nodes, reflected by the fact that the polynomial matrix $A(\q)$ is restricted to be symmetric. \QED
\end{remark}

\subsection{Consistency}
Now, we can formulate the consistency result as follows.
\begin{theorem}[Consistency]\label{th:con}
  Consider a data generating system $\mathcal{S}$ as defined in Definition \ref{def:dgn} and a model set $\mathcal{M}$ in which the predictor filter \eqref{eq:pf} is uniformly stable. Then $M(\tn)$ is a consistent estimate of $M^0$ if the following conditions hold:
  \begin{enumerate}
    \item The true system is in the model set ($\mathcal{S}\in\mathcal{M}$).
    \item The data are informative with respect to the model set.
    \item The model set is globally network identifiable.
  \end{enumerate}
\end{theorem}
\begin{proof}
  The proof consists of three steps. First, convergence of $V_N(\theta)$ to $\bar{V}(\theta)$ for $N\rightarrow\infty$ follows directly from applying Theorem 2B.1 of \cite{LJUNG1999} and the fact that $S\succ0$ as the conditions for convergence are satisfied by the network model set. Second, by Condition 1, $\theta^0$ is a minimum of $\bar{V}(\theta)$, which can be seen as follows. As $r(t)$ and $e(t)$ are uncorrelated and $W_{\bar{\varepsilon}e}(q,\theta)$ is strictly proper, the power of any cross term between the three terms in the prediction error \eqref{eq:pere} is zero, so the power of each term can be minimized individually. As a result, $W_{\bar{\varepsilon}r}(q,\theta^0)=0$ and $W_{\bar{\varepsilon}e}(q,\theta^0)=0$ and thus the cost function reaches its minimum value when the prediction error is equal to the innovation as in \eqref{eq:pein}. Third, following the result of Theorem 8.3 in \cite{LJUNG1999}, under Condition 2, this minimum of $\bar{V}(\theta)$ at $\theta^0$ provides a unique predictor filter $W(q,\theta)$ and therefore also a unique pair $(T_{wr}(q),\Phi_{\bar v}(\omega))$. With Condition 3 this implies that the resulting model $M(\theta)=M(\theta^0)$ is unique. Therefore, $M(\tn)$ converges to $M(\theta^0)$ with probability 1.
\end{proof}
Observe that any weight $S\succ0$ leads to consistent estimates, but that minimum variance is only achieved for $S=\bar{\Lambda}^{-1}$.

Now it has been proved that diffusively coupled networks can be identified consistently, the next step is to formulate algorithms for obtaining these estimates.

\section{A multi-step algorithm} \label{sec:alg}
The parametrized prediction error \eqref{eq:pe} is not affine in the parameters $\theta$. Only in the very special situation where $F(q,\theta)=I$ and $A_0(\theta)=I$, the structure of \eqref{eq:pe} is affine. This situation causes the optimization problem \eqref{eq:crit} to be nonconvex. Especially for networks with many nodes, this results in high computational complexity and occurrences of local optima. One approach to reduce the problem is to solve multiple multi-input single-output (MISO) problems instead of one large multi-input multi-output (MIMO) problem \cite{VANDENHOF2013,DANKERS2016,GEVERS2018}. However, since the dynamics are coupled (that is, $A(\q)$ is symmetric and therefore its elements are not independently parametrized), a decomposition into MISO problems cannot be made without loss of accuracy.

In this section, as an alternative, a multi-step algorithm is developed, where in each step a quadratic problem is solved using a linear regression scheme. With that, the developed method contains steps that are similar to sequential least squares (SLS) \cite{WEERTS2018_I}, weighted null-space fitting (WNSF) \cite{GALRINHO2019}, and the multi-step least squares method in \cite{FONKEN2021}, but particularly tuned to the network model structure of the current paper.

In Step 1 an unstructured nonparametric ARX model is estimated from data. This ARX model is reduced to a structured parametric ARMAX network model in Step 2, which is improved in Step 3. Once the network dynamics have been estimated, the noise model is found in Step 4 and the discrete-time components are extracted in Step 5. Finally, the components are mapped back to the continuous-time domain in Step 6. The particular difference of our method with the aforementioned methods is that the structure in the parametrized objects in Step 2 and Step 3 is different and that the optimization problem in Step~2 is a constrained optimization problem. Furthermore, Steps 4, 5, and 6 have been added.

As only quadratic problems are solved, the optimizations are convex and have unique solutions. In this way, the formulated algorithm achieves a consistent parameter estimation with minimum variance and limited computation complexity. This makes the algorithm also applicable to networks with many nodes.

For the development of the algorithm we will restrict attention to the situation of an ARMAX-like model structure, where we consider a data generating system $\mathcal{S}=(A^0,B^0,F^0,\Lambda^0)$ with $F^0(q):=C^0(\q)$ being a monic polynomial, leading to the network equation
\begin{equation}\label{eq:armax}
     A^0(\q) w(t) = B^0(\q) r(t) + C^0(\q) e(t),
\end{equation}
which would have an ARMAX structure if $A^0(\q)$ would be monic. Multiplying both sides of \eqref{eq:armax} with $[C^0(\q)A_0^0]^{-1}$ leads to
\begin{equation}\label{eq:arxInf}
  \breve{A}^0(\q) w(t) = \breve{B}^0(\q) r(t) + \bar{e}(t),
\end{equation}
where $\breve{A}^0(\q)$ is monic, $\bar{e}(t)$ is the innovation \eqref{eq:in}, and
\begin{align}
  \breve{A}^0(\q) &= [\bar{C}^0(\q)]^{-1}A^0(\q),   \label{eq:relAC1}\\
  \breve{B}^0(\q) &= [\bar{C}^0(\q)]^{-1}B^0(\q),   \label{eq:relBC1}\\
  \bar{C}^0(\q) &= C^0(\q)A^0_0.                    \label{eq:relCC1}
\end{align}
Now consider the model structure $A(\q,\theta_a)$, $B(\q,\theta_b)$, and $\bar{C}(\q,\eta_c,\theta_a)$, as models of $A^0(\q)$, $B^0(\q)$, and $\bar{C}^0(\q)$, respectively, with $\bar{C}(\q,\eta_c,\theta_a) = C(\q,\theta_c)A_0(\theta_a)$, and with $\vartheta := \begin{bmatrix} \theta_a^{\top} & \theta_b^{\top} & \eta_c^{\top} \end{bmatrix}^{\top}$.As $C(\q)$ is monic and $A_0$ is constant, $\bar{C}_0=A_0$ and therefore parametrized as such. All other matrices $\bar{C}_{\ell}$ are independently parametrized with parameters $\eta_c$. The exact parametrization is given in Appendix~\ref{app:struct}.

\subsection*{Step 1: Estimating the nonparametric ARX model}
As a first step, we are going to estimate a nonparametric ARX model for \eqref{eq:arxInf}, by parametrizing the infinite series expansions $\breve{A}^0(\q)$ and $\breve{B}^0(\q)$ by high order polynomial (finite) expansions $\breve{A}(\q,\zeta^n)$ and $\breve{B}(\q,\zeta^n)$, according to
\begin{align}
  \bar\varepsilon_A(t,\zeta^n)
  &= \breve{A}(\q,\zeta^n)w(t) - \breve{B}(\q,\zeta^n)r(t), \label{eq:peARX1}\\
  &= w(t) - [\varphi^n(t)]^{\top}\zeta^n,                   \label{eq:peARX2}
\end{align}
with $n$ the finite order of the polynomials, which is typically chosen to be high. The parameter vector $\zeta^n$ and the matrix $[\varphi^n(t)]^{\top}$ are given in Appendix \ref{app:arx}. The nonparametric ARX model \eqref{eq:arxInf} is then estimated through estimating its parameters $\zeta^n$. As this step serves to make an initial estimate of the network, the network structure is not taken into account. Further, consistency of this step is only achieved if the order $n$ tends to infinity as function of the data length $N$ at an appropriate rate, according to \cite{LJUNG1992}. However, the bias will be negligibly small if the order $n$ is chosen sufficiently large. The least-squares estimate of $\zeta^n$ is found by
\begin{equation}\label{eq:ZetaH}
    \zn = \left[\frac{1}{N}\sum_{t=n+1}^{N} \varphi^n(t)[\varphi^n(t)]^{\top}\right]^{-1} \left[\frac{1}{N}\sum_{t=n+1}^{N} \varphi^n(t)w(t)\right].
\end{equation}
Under conditions of consistent estimation, and so if $n$ and $N$ approach infinity, $\bar\varepsilon_A(t,\zn)$ will be an accurate estimate of the innovation $\bar{e}(t)$.
The covariance of the innovation is estimated as the covariance of the residual as
\begin{equation}\label{eq:Cov0}
  \bar\Lambda(\zn) = \frac{1}{N}\sum_{t=n+1}^{N} \bar\varepsilon_A(t,\zn) \bar\varepsilon_A^{\top}(t,\zn),
\end{equation}
with residual \eqref{eq:peARX1} evaluated at $\zn$. The covariance of the parameter estimation error $\epsilon(t,\zn) := \zn - \zo$ with $\zo$ the actual coefficients of the expansions in \eqref{eq:arxInf}, is estimated by
\begin{equation}\label{eq:CovE1}
  P(\zn) = \left[ \frac{1}{N}\sum_{t=n+1}^{N} \varphi^n(t) \bar\Lambda^{-1}(\zn) [\varphi^n(t)]^{\top} \right]^{-1}.
\end{equation}
\begin{remark}\label{rem:miso}
  As each row in \eqref{eq:peARX2} is independently parametrized, the parameters $\zeta^n$ can be estimated for each row independently, resulting in $L$ MISO problems instead of one MIMO problem. This is attractive for networks with many nodes. \QED
\end{remark}

\subsection*{Step 2: Reducing to the structured network model}
The high order ARX model is used to identify the structured network model through the relations \eqref{eq:relAC1} and \eqref{eq:relBC1}. In this step, the structural properties of $A^0(\q)$ are incorporated and the parameter constraint is taken into account to fix the scaling parameter and obtain a unique solution.

The relations \eqref{eq:relAC1} and \eqref{eq:relBC1} are equivalently written as
\begin{align}
  A^0(\q) - \bar{C}^0(\q) \breve{A}^0(\q) &= 0,     \label{eq:relAC2}\\
  B^0(\q) - \bar{C}^0(\q) \breve{B}^0(\q) &= 0.     \label{eq:relBC2}
\end{align}
Then from \eqref{eq:relAC2} and \eqref{eq:relBC2} we can extract:
\begin{equation}\label{eq:reg}
  - Q(\zo)\vartheta^0 = 0,
\end{equation}
where $\vartheta^0$ represents the coefficients of the actual underlying system described by $A^0(\q)$, $B^0(\q)$, and $\bar{C}^0(\q)$ \eqref{eq:armax} and where the nonparametric ARX representation $\breve{A}^0(\q)$ and $\breve{B}^0(\q)$ of the system is incorporated in $Q(\zo)$. The specific structure of $Q(\zo)$ is provided in Appendix~\ref{app:QT}. The polynomial terms in \eqref{eq:relAC2} and \eqref{eq:relBC2} are considered up to time lag $n$, and the row dimension of $Q(\zo)$ is equal to $\text{dim}(\zo)$.

On the basis of the estimated nonparametric ARX model parameters $\zn$, an initial least-squares\footnote{Weighted least-squares can be used as well (see Step 3) with weighting matrix $W(\zn) = P^{-1}(\zn)$ \cite{GALRINHO2016_T}.} estimate of $\vartheta^0$ is obtained by the linear constrained optimization problem
\begin{align}\label{eq:opt1}
  \vtn^{(0)} = \min_{\vartheta} &\quad \vartheta^{\top}Q^{\top}(\zn) Q(\zn)\vartheta \\
  \text{subject to} & \quad \Gamma\vartheta=\gamma,\label{eq:constr1}
\end{align}
where the constraint \eqref{eq:constr1} results from Condition 4 in Proposition \ref{prop:idf}. The optimization problem can be solved using the Lagrangian and the Karush–Kuhn–Tucker conditions \cite{CHONG2008}, giving
\begin{equation}\label{eq:Theta1}
  \begin{bmatrix}\vtn^{(0)}\\ \hat\lambda_N^{(0)} \end{bmatrix} = \begin{bmatrix}
  Q^{\top}(\zn)Q(\zn) & \Gamma^{\top}\\ \Gamma&0 \end{bmatrix}^{-1} \begin{bmatrix}0\\ \gamma\end{bmatrix},
\end{equation}
where $\hat\lambda_N^{(0)}$ are the estimated Lagrange multipliers. The covariance of the residuals is updated according to (initially for $k=0$):
\begin{equation}\label{eq:Cov1}
  \bar\Lambda(\vtn^{(k)}) = \frac{1}{N}\sum_{t=n+1}^{N} \bar\varepsilon(t,\vtn^{(k)}) \bar\varepsilon^{\top}(t,\vtn^{(k)}),
\end{equation}
with residual
\begin{multline}
  \bar\varepsilon(t,\vtn^{(k)}) = \bar{C}^{-1}(\q,\vtn^{(k)}) \\ \left[ A(\q,\vtn^{(k)}) w(t) -  B(\q,\vtn^{(k)}) r(t) \right].
\end{multline}

\subsection*{Step 3: Improving the structured network model}
This step aims to correct for the residuals in \eqref{eq:reg} that are not accounted for in \eqref{eq:Theta1}, due to the fact that only a high order approximation of the nonparametric ARX model is used.

Substituting $\breve{A}(\q,\zn)$ and $\breve{B}(\q,\zn)$ for $\breve{A}^0(\q)$ and $\breve{B}^0(\q)$, respectively, into \eqref{eq:relAC2} and \eqref{eq:relBC2} gives
\begin{align}
  &A^0(\q) - \bar{C}^0(\q) \breve{A}^0(\q) = \nonumber \\
  &\hspace{8em} \bar{C}^0(\q) [\breve{A}(\q,\zn)-\breve{A}^0(\q)],  \label{eq:relAC} \\
  &B^0(\q) - \bar{C}^0(\q) \breve{B}^0(\q) = \nonumber\\
  &\hspace{8em} \bar{C}^0(\q) [\breve{B}(\q,\zn)-\breve{B}^0(\q)],  \label{eq:relBC}
\end{align}
which are equivalently written as (by using \eqref{eq:reg})
\begin{equation}\label{eq:reg2}
  - Q(\zn)\vartheta^0 = T(\vartheta^0)(\zn-\zeta^{no}),
\end{equation}
where the matrix $T(\vartheta^0)$ is given in Appendix \ref{app:QT}. The estimate of $\vartheta^0$ with minimum variance is obtained by solving a weighted least-squares problem, where the weighting matrix is given by the inverse covariance matrix of the right-hand side expression in \eqref{eq:reg2}. As this term depends on $\vartheta$ this problem is solved iteratively by
\begin{align}\label{eq:opt2}
  \vtn^{(k)} = \min_{\vartheta} &\quad \vartheta^{\top}Q^{\top}(\zn) W(\vtn^{(k-1)}) Q(\zn)\vartheta \\
  \text{subject to} & \quad \Gamma\vartheta=\gamma,\label{eq:constr2}
\end{align}
where the weighting matrix $W(\vtn^{(k-1)})$ is iteratively updated for $k = 1,2,\cdots$ according to
\begin{align}\label{eq:Weight2}
  & W(\vtn^{(k-1)}) = T^{-\top}(\vtn^{(k-1)}) P^{-1}(\vtn^{(k-1)}) T^{-1}(\vtn^{(k-1)}), \\
  & \hspace{-2mm}\text{where $P(\vtn^{(k-1)})$ is updated according to} \nonumber \\
  & P^{-1}(\vtn^{(k-1)}) = \frac{1}{N}\sum_{t=n+1}^{N} \varphi^n(t) \bar\Lambda^{-1}(\vtn^{(k-1)}) [\varphi^n(t)]^{\top}.
\end{align}
Similar to Step~2, this optimization problem can be solved through
\begin{equation}\label{eq:Theta2}
  \begin{bmatrix}\vtn^{(k)}\\ \hat\lambda_N^{(k)} \end{bmatrix} = \begin{bmatrix}
  Q^{\top}(\zn)W(\vtn^{(k-1)})Q(\zn) & \Gamma^{\top}\\ \Gamma&0 \end{bmatrix}^{-1} \begin{bmatrix}0\\ \gamma\end{bmatrix},
\end{equation}
where $\hat\lambda_N^{(k)}$ are the estimated Lagrange multipliers. Finally, the covariance of the residuals is updated according to \eqref{eq:Cov1}.

\begin{remark}
  Although this step is asymptotically efficient, iterating may improve the estimate for finite data length $N$. The cost
  \begin{equation}\label{eq:cost_det}
    V_N(\vtn^{(k)}) = \frac{1}{N} \det \sum_{t=1}^{N} \bar{\varepsilon}(t,\vtn^{(k)}) \bar{\varepsilon}^{\top}(t,\vtn^{(k)})
  \end{equation}
  is evaluated at each iteration to decide whether the parameter estimation has improved. However, as \eqref{eq:cost_det} is not affine in the parameters, an improved cost may still result in deteriorated parameter estimates. The cost \eqref{eq:cost_det} is used as it is independent of $\Lambda(\theta)$ and under Gaussian assumptions, minimizing \eqref{eq:cost_det} results in minimum variance of the estimates if $\Lambda(\theta)$ is independently parametrized from $A(\q,\theta)$, $B(\q,\theta)$, and $C(\q,\theta)$. In this situation, the asymptotic (minimum) variance resulting from \eqref{eq:costconv} is equal to the asymptotic variance of the maximum likelihood estimator \cite{LJUNG1999}. \QED
\end{remark}

\subsection*{Step 4: Obtaining the noise model}
Having estimated $A^0(\q)$, $B^0(\q)$, and $\bar{C}^0(\q)$, the noise model represented by $C^0(\q)$ and $\Lambda^0$ can be recovered. On the basis of \eqref{eq:relCC1}, the estimate of $C^0(\q)$ is constructed as
\begin{equation}\label{eq:ThetaCh}
  C(\q,\tn^{(k)}) = \bar{C}(\q,\vtn^{k})A_0^{-1}(\vtn^{k}).
\end{equation}
Further, as $\Lambda^0 = A_0^0\bar\Lambda^0A_0^0$, the estimate of $\Lambda^0$ is given by
\begin{equation}\label{eq:ThetaLh}
  \Lambda(\tn^{(k)}) = A_0(\vtn^{(k)})\bar{\Lambda}(\vtn^{(k)})A_0(\vtn^{(k)}).
\end{equation}

\subsection*{Step 5: Estimating the discrete-time components}
With the estimate of $A^0(\q)$ from Step~3, the dynamics of the discrete-time network have been estimated. The components, represented by $\bar{X}(\q)$ and $\bar{Y}(\q)$, are obtained through the inverse mapping of $A(\q) = \bar{X}(\q) + \bar{Y}(\q)$, with $\bar{X}(\q)$ diagonal and $\bar{Y}(\q)$ Laplacian, given by
\begin{align}
    \bar{x}_{jk,\ell} &= \begin{cases}
    0                            & \mbox{if } j\neq k \\
    a_{jj,\ell} + \sum_{i\neq j} a_{ij,\ell},  & \mbox{if } j=k \end{cases} \label{eq:Xbar}\\
    \bar{y}_{jk,\ell} &= \begin{cases}
    a_{jk,\ell},                 & \mbox{if } j\neq k \\
    -\sum_{i\neq j} a_{ij,\ell}, & \mbox{if } j=k. \end{cases} \label{eq:Ybar}
\end{align}

\subsection*{Step 6: Estimating the continuous-time components}
The continuous-time representation $X(p), Y(p), \bar B(p)$ can be obtained from the estimated discrete-time model from Steps 5 and 3 through the inverse mapping of \eqref{eq:bar1}-\eqref{eq:bar2}, given by
\begin{align}
    x_{jk,\ell}& = (-T_s)^{\ell} \textstyle\sum_{i=\ell}^{n_a} \binom{i}{\ell} \bar{x}_{jk,i},\label{eq:X}\\
    y_{jk,\ell}& = (-T_s)^{\ell} \textstyle\sum_{i=\ell}^{n_a} \binom{i}{\ell} \bar{y}_{jk,i},\label{eq:Y}\\
    \bar{b}_{jj,\ell}& = (-T_s)^{\ell} \textstyle\sum_{i=\ell}^{n_b} \binom{i}{\ell}  b_{jj,i}.\label{eq:barB}
\end{align}

\subsection*{The complete algorithm}
The above steps describe the procedure for identifying the physical components of a diffusively coupled linear network with an ARMAX-like model structure. This procedure leads to the following algorithm.
\begin{algorithm}[ARMAX-like model structure]\label{alg:ARMAX}
  Consider a data generating system $\mathcal{S}$ with $F^0(q):=C^0(\q)$ a monic polynomial and a network model set $\mathcal{M}$ \eqref{eq:modelset} with $F(q,\theta):= C(\q,\theta)$ a monic polynomial. Then $M(\hat\theta_N)$, a consistent estimate of $M^0$, is obtained through the following steps:
  \begin{enumerate}
    \item Estimate the nonparametric ARX model \eqref{eq:arxInf} by least squares \eqref{eq:ZetaH} to obtain $\zn$.
    \item Reduce the nonparametric ARX model to a parametric model \eqref{eq:ABF_DT} by weighted least-squares \eqref{eq:Theta1} to obtain $\vtn^{(0)}$.
    \item Improve the parametric model \eqref{eq:ABF_DT} by weighted least-squares \eqref{eq:Theta1} to obtain $\vtn^{(k)}$ for $k=1,2,\ldots$.
    \item Obtain the noise model by calculating \eqref{eq:ThetaCh} and \eqref{eq:ThetaLh} to obtain $C(\q,\hat\theta_N^{(k)})$ and $\Lambda(\hat\theta_N^{(k)})$.
    \item Obtain the discrete-time component values through \eqref{eq:Xbar} and \eqref{eq:Ybar} to estimate $\bar{X}(\q)$ and $\bar{Y}(\q)$.
    \item Obtain the continuous-time parametric model through \eqref{eq:X}-\eqref{eq:barB} to estimate $X(p)$, $Y(p)$, and $\bar{B}(p)$. \QED
  \end{enumerate}
\end{algorithm}
Consistency and minimum variance of the estimates obtained with Algorithm \ref{alg:ARMAX} follows from the similarity with WNSF and its proof, under technical conditions on the rates with which $n$ and $N$ tend to infinity \cite{GALRINHO2019}. The main difference is that $A(\q,\theta)$ is nonmonic and symmetrically parametrized, resulting in a different structure in \eqref{eq:reg}. In particular, the structure in $Q(\zeta^n)$ and $T(\vartheta^0)$ is different and the optimization problem \eqref{eq:opt1}-\eqref{eq:constr1} is constrained. For consistency, $Q(\zn)$ needs to have full column rank, which can be shown to be satisfied if the identifiability conditions in Proposition \ref{prop:idf} are satisfied. Consistency of Step~4, 5, and 6 follows naturally.
\begin{remark}
  In order to perform Algorithm \ref{alg:ARMAX}, the measured data $\{z(t)\}$ is needed; the order $n$ of the ARX model needs to be chosen; and the true orders $n_a$, $n_b$, and $n_c$ of $A(\q)$, $B(\q)$, and $C(\q)$, respectively, need to be known. \QED
\end{remark}
\begin{remark}[Simplification to an ARX-like model structure]
  If the noise is not filtered, that is $F(q):=C(\q)=I$, the network has an ARX-like model structure and the ARX model \eqref{eq:arxInf} can exactly describe the diffusively coupled network, where $\breve{A}(\q)$ and $\breve{B}(\q)$ are of the same order as $A(\q)$ and $B(\q)$, respectively. Algorithm \ref{alg:ARMAX} improves in the sense that Step 1 is consistent for sufficiently large data length $N$ and therefore, no additional estimation error is made in Step 2, which makes Step 3 superfluous. \QED
\end{remark}
\begin{remark}[Simplification to an ARX model structure]
  If $A_0=I$ in addition to unfiltered noise ($F(q):=C(\q)=I$), the network has an ARX model structure. In this case, the network can consistently be identified in a single step, by incorporating the symmetric structure in Step 1 of Algorithm \ref{alg:ARMAX} and by choosing the order of $\breve{A}(\q)$ and $\breve{B}(\q)$ equivalent to the order of $A(\q)$ and $B(\q)$, respectively. The resulting identification procedure has been described in \cite{KIVITS2019}. \QED
\end{remark}

\begin{figure}[t]
  \centering
  \includegraphics[width=0.95\columnwidth]{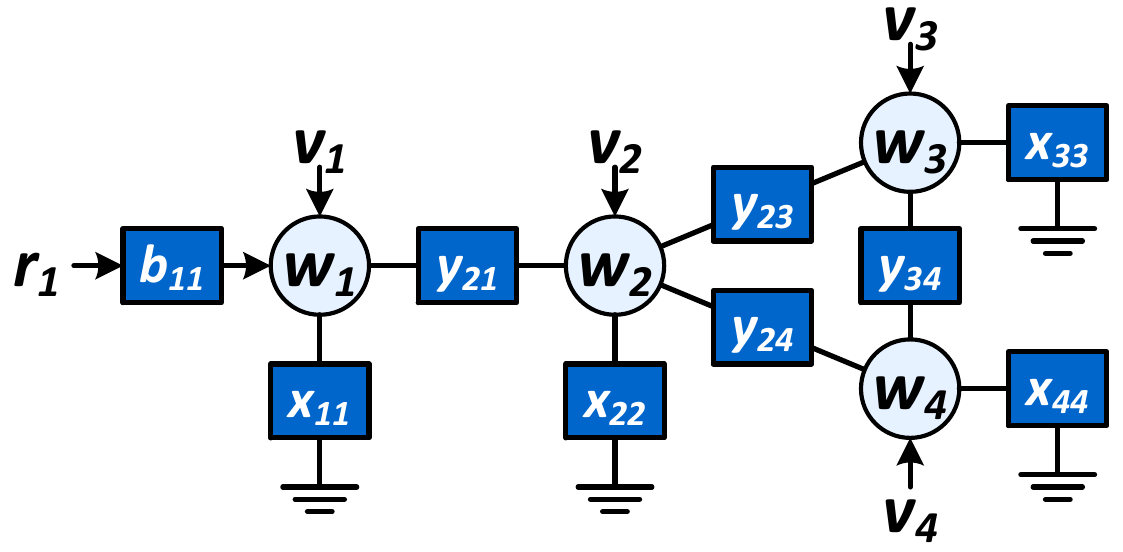}
  \caption{The continuous-time network model with interconnection dynamics described by the polynomials $y_{jk}(p)$, dynamics to the ground described by the polynomial $x_{jj}(p)$, and static excitation filter $b_{11}$. }
  \label{fig:sim}
\end{figure}

\begin{table}[t]
    \centering
    \caption{The order $n$ of the ARX fit, the data length $N$, and the rate $n^4/N$ for each set of Experiment 1. }
    \setlength\tabcolsep{3pt}
    \resizebox{\hsize}{!}{\begin{tabular}{l||c|c|c|c|c|c|c|c|c|c}
          Set & 1     & 2     & 3     & 4     & 5     & 6     & 7     & 8     & 9     & 10 \\ \hline \hline
          $n$ & 3     & 4     & 5     & 6     & 7     & 8     & 9     & 10    & 11    & 12    \\ \hline
          $N$ & 96    & 320   & 834   & 1852  & 3694  & 6827  & 11930 & 20000 & 32536 & 51841 \\ \hline
      $n^4/N$ & 0.85  & 0.80  & 0.75  & 0.70  & 0.65  & 0.60  & 0.55  & 0.50  & 0.45  & 0.40  \\
    \end{tabular}}
    \label{tab:exp1}
\end{table}

\begin{table*}[t]
\begin{center}
  \caption{The true parameter values of $X(p)$ and $Y(p)$ and the mean and standard deviation (SD) of their estimates for both sets ($K=1$ and $K=3$) of Experiment 2. }
  \setlength\tabcolsep{3pt}
  \resizebox{\hsize}{!}{\begin{tabular}{l|| c|c|c|c|c|c|c|c}
    Parameter   & $\theta^c_1$&$\theta^c_2$&$\theta^c_3$&$\theta^c_4$&$\theta^c_5$&$\theta^c_6$&$\theta^c_7$&$\theta^c_8$ \\ \hline
    True value  & $1$&$0$&$1\times10^{-2}$&$1$&$0$&$2\times10^{-2}$&$1$&$0$\\ \hline \hline
    $K=1$ Mean  & $1.0136$&-$5.6440\times10^{-4}$&$9.9999\times10^{-3}$&$1.0171$&-$7.3562\times10^{-4}$&$2.0054\times10^{-2}$&$1.0221$&-$1.6503\times10^{-3}$\\ \hline
    $K=1$ SD    & $3.6467\times10^{-2}$&$2.0287\times10^{-3}$&$2.3009\times10^{-6}$&$9.5878\times10^{-2}$&$6.1179\times10^{-3}$&$1.3032\times10^{-4}$&$1.4998\times10^{-1}$&$8.2037\times10^{-3}$\\ \hline \hline
    $K=3$ Mean  & $1.0012$&-$7.7029\times10^{-5}$&$1.0000\times10^{-2}$&$1.0032$&-$2.0946\times10^{-4}$&$2.0001\times10^{-2}$&-$1.0085$&-$5.6775\times10^{-4}$\\ \hline
    $K=3$ SD    & $4.7774\times10^{-3}$&$2.8210\times10^{-4}$&$1.2971\times10^{-6}$&$8.2831\times10^{-3}$&$4.9076\times10^{-4}$&$1.3605\times10^{-5}$&$1.5390\times10^{-2}$&$1.1072\times10^{-3}$\\ \hline \hline \hline
    Parameter   & $\theta^c_9$&$\theta^c_{10}$&$\theta^c_{11}$&$\theta^c_{12}$&$\theta^c_{13}$&$\theta^c_{14}$&$\theta^c_{15}$&$\theta^c_{16}$ \\ \hline
    True value  & $5\times10^{-2}$&$10$&$0$&$7\times10^{-2}$& -$4$&-$3\times10^{-1}$&$0$&$0$\\ \hline \hline
    $K=1$ Mean  & $4.9848\times10^{-2}$&$9.9676$&$1.4936\times10^{-3}$&$6.9868\times10^{-2}$&$-4.0061$&-$3.0078\times10^{-1}$&$6.2954\times10^{-3}$&$2.1455\times10^{-4}$\\ \hline
    $K=1$ SD    & $8.5376\times10^{-4}$&$1.5783\times10^{-1}$&$8.0059\times10^{-3}$&$6.6951\times10^{-4}$&$5.4752\times10^{-2}$&$1.8156\times10^{-3}$&$6.6368\times10^{-2}$&$6.6894\times10^{-4}$\\ \hline \hline
    $K=3$ Mean  & $4.9992\times10^{-2}$&$1.003\times10^{1}$&-$1.8345\times10^{-4}$&$7.008\times10^{-2}$&-$3.9998$&-$3.0003\times10^{-1}$&$4.2531\times10^{-4}$&-$2.4708\times10^{-5}$\\ \hline
    $K=3$ SD    & $4.8096\times10^{-5}$&$3.0088\times10^{-2}$&$1.8316\times10^{-3}$&$7.9796\times10^{-5}$&$2.8000\times10^{-3}$&$1.9475\times10^{-4}$&$3.0851\times10^{-3}$&$1.1680\times10^{-4}$\\ \hline \hline \hline
    Parameter   & $\theta^c_{17}$&$\theta^c_{18}$&$\theta^c_{19}$&$\theta^c_{20}$&$\theta^c_{21}$&$\theta^c_{22}$&$\theta^c_{23}$&$\theta^c_{24}$ \\ \hline
    True value  & $0$&$0$&$0$&-$4\times10^{-1}$& -$8$&$0$&-$9$&-$6\times10^{-1}$ \\ \hline \hline
    $K=1$ Mean  & $8.9935\times10^{-3}$&-$1.0261\times10^{-4}$&$1.3912\times10^{-2}$&-$4.0102\times10^{-1}$&-$8.0188$&-$3.3463\times10^{-4}$&-$8.9442$&-$5.9889\times10^{-1}$\\ \hline
    $K=1$ SD    & $3.8590\times10^{-2}$&$6.2002\times10^{-4}$&$1.6645\times10^{-1}$&$9.1170\times10^{-3}$&$1.1026\times10^{-1}$&$3.9220\times10^{-3}$&$1.5367\times10^{-1}$&$1.2576\times10^{-2}$\\ \hline \hline
    $K=3$ Mean  & $5.2077\times10^{-4}$&-$2.0221\times10^{-5}$&$2.1762\times10^{-3}$&-$4.0005\times10^{-1}$&-$8.0012$&-$5.8232\times10^{-5}$&-$8.9958$&-$6.0024\times10^{-1}$\\ \hline
    $K=3$ SD    & $3.4376\times10^{-3}$&$1.6190\times10^{-4}$&$6.1368\times10^{-3}$&$3.9797\times10^{-4}$&$7.6248\times10^{-3}$&$2.8442\times10^{-4}$&$1.3217\times10^{-2}$&$9.3302\times10^{-4}$\\
  \end{tabular}}
  \label{tab:parct}
\end{center}
\end{table*}

\begin{figure}[t]
  \centering
  \includegraphics[width=\columnwidth]{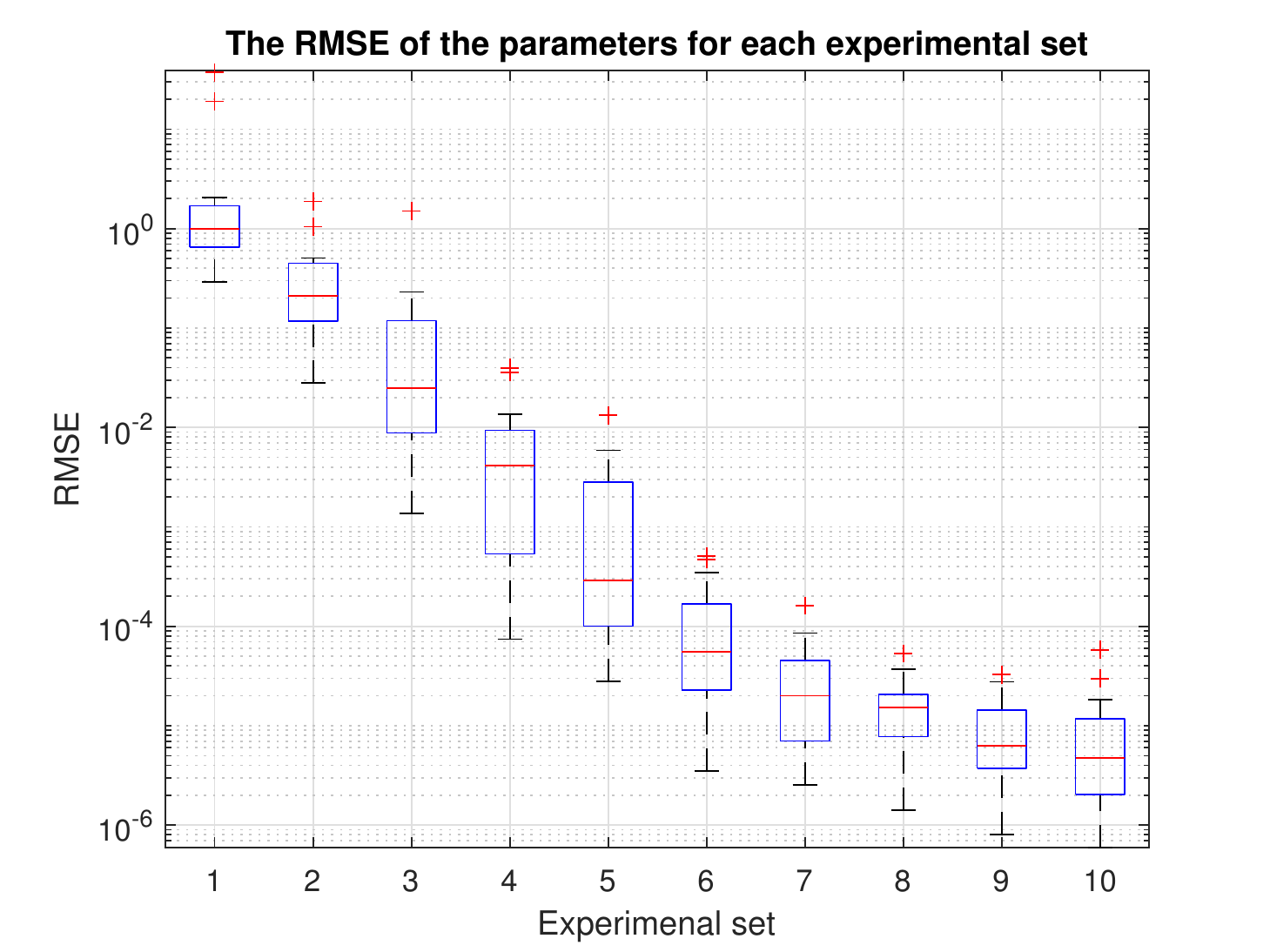}
  \caption{Boxplot of the relative means squared error (RMSE) \eqref{eq:rmse} of \ \ \linebreak  the parameters of $X(p)$ and $Y(p)$ for each experimental set. }
  \label{fig:rmse}
\end{figure}

\begin{figure*}[t]
  \centering
  \begin{minipage}[t]{.5\textwidth}
    \centering
    \includegraphics[width=\columnwidth]{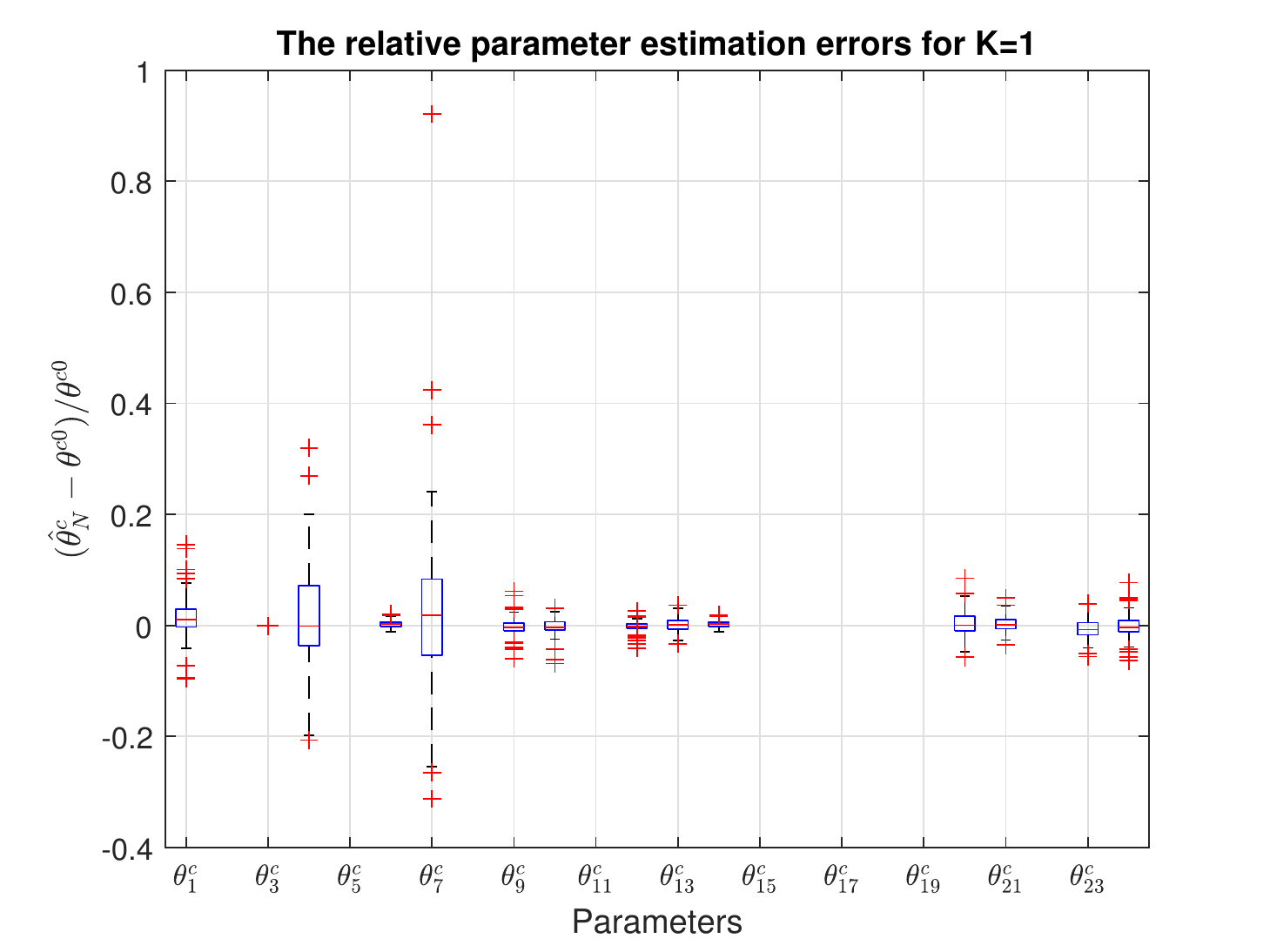}
    \caption{Boxplot of the relative estimation errors of the parameters of $X(p)$ \ \ \ \linebreak and $Y(p)$ for the experimental set with a single excitation signal ($K=1$), \ \ \ \linebreak for parameters with a nonzero true value. }
    \label{fig:rel1}
  \end{minipage}%
  \begin{minipage}[t]{.5\textwidth}
    \centering
    \includegraphics[width=\columnwidth]{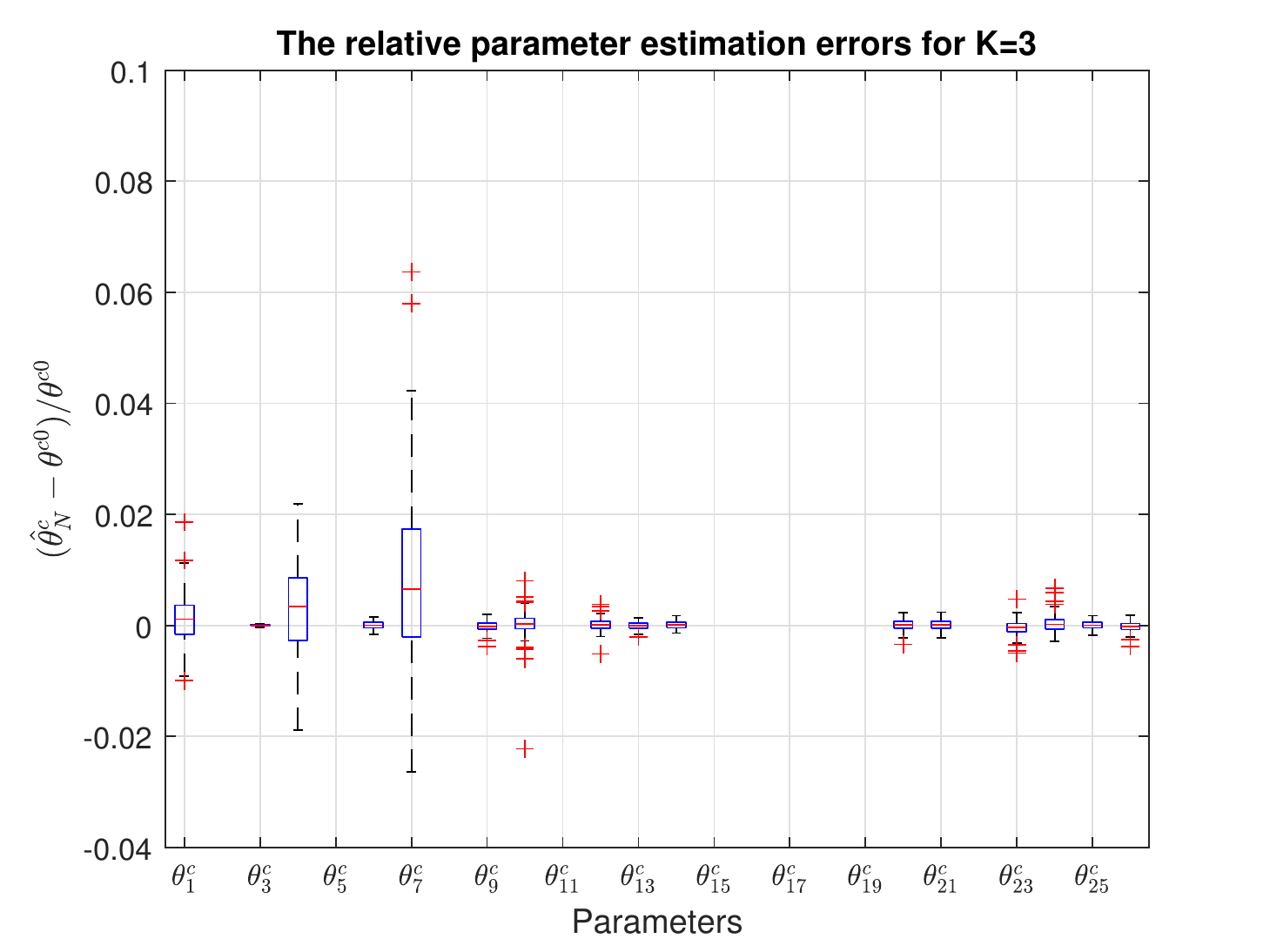}
    \caption{Boxplot of the relative estimation errors of the parameters of $X(p)$, \ \ \linebreak $Y(p)$ and $B(p)$ for the experimental set with three excitation signals ($K=3$), \ \ \ \linebreak for parameters with a nonzero true value. }
    \label{fig:rel3}
  \end{minipage}\\
  \begin{minipage}[t]{.5\textwidth}
    \centering
    \includegraphics[width=\columnwidth]{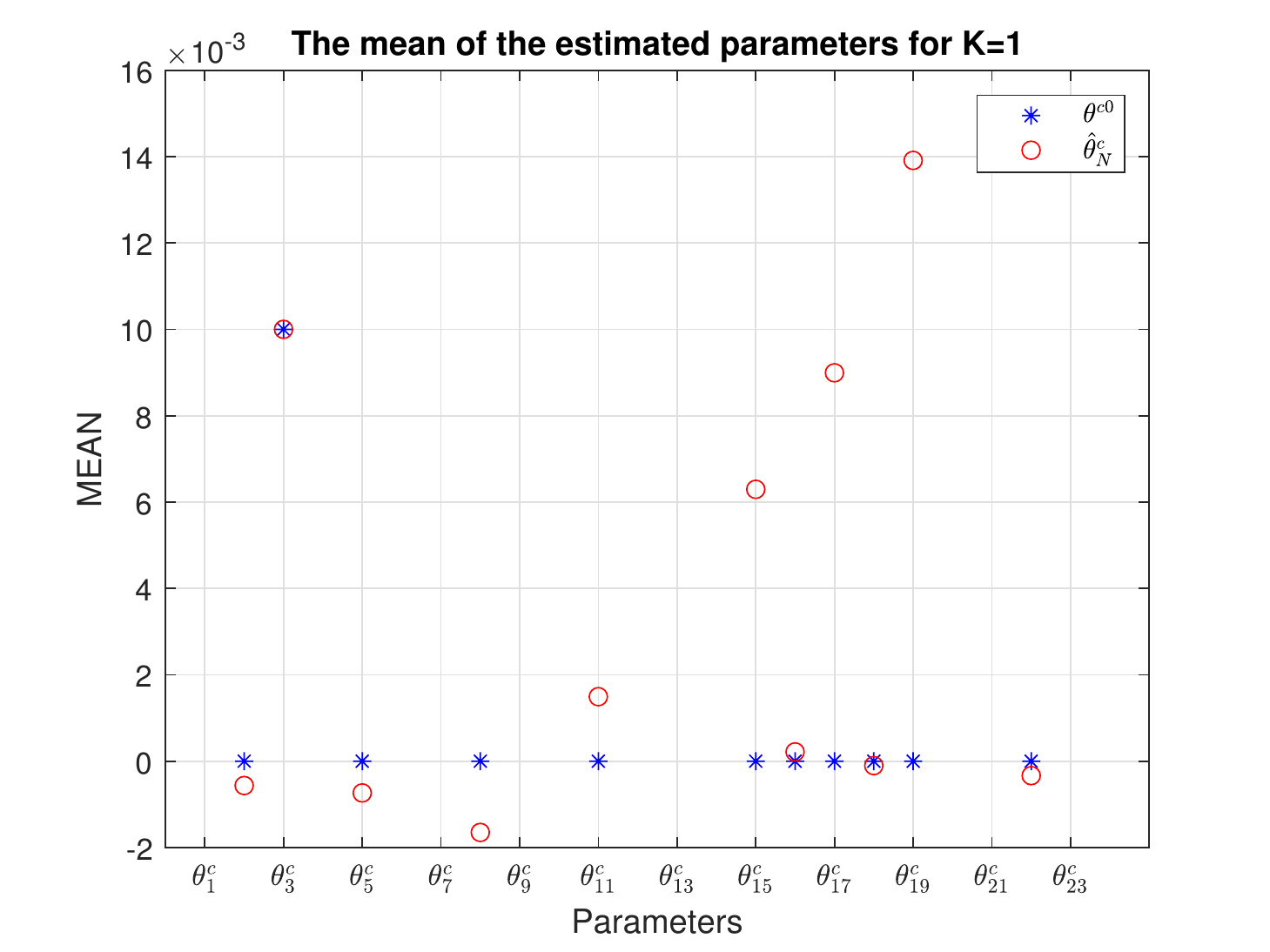}
    \caption{The true parameter values (blue) and the mean of the estimated \ \ \ \linebreak parameter values (red) of $X(p)$ and $Y(p)$ for the experimental set with \ \ \ \linebreak a single excitation signal ($K=1$), focusing on parameters with a true \ \ \ \linebreak value of $0$. }
    \label{fig:mean1}
  \end{minipage}%
  \begin{minipage}[t]{.5\textwidth}
    \centering
    \includegraphics[width=\columnwidth]{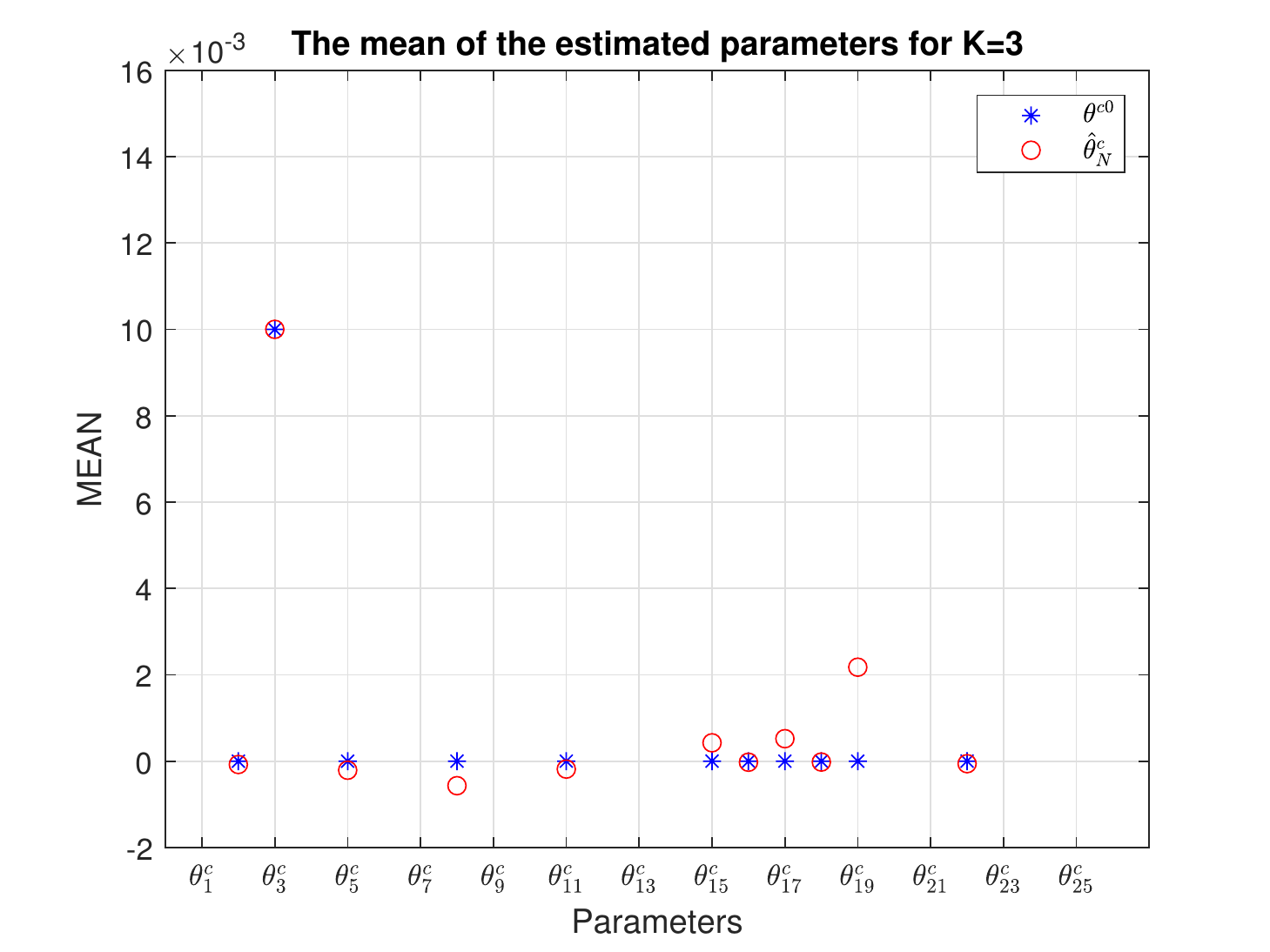}
    \caption{The true parameter values (blue) and the mean of the estimated \ \ \ \linebreak parameter values (red) of $X(p)$, $Y(p)$ and $B(p)$ for the experimental set \ \ \ \linebreak with three excitation signals ($K=3$), focusing on parameters with a true \ \ \  \linebreak value of $0$. }
    \label{fig:mean3}
  \end{minipage}
\end{figure*}

\section{Simulation example}\label{sec:simex}
This section contains a simulation example that serves to illustrate the theory and to show that indeed the topology and the physical components of a network can be identified using a single excitation signal only. The identification is performed with the algorithm presented above.

\subsection{Experimental set-up}
Consider the continuous-time diffusively coupled network \eqref{eq:xy} consisting of four one-dimensional nodes, with external signal $u(t) = B_0r(t)+v(t)$, described by
\begin{multline}\label{eq:simct}
  (X_0+Y_0) w(t) + (X_1+Y_1) \frac{d}{dt}w(t) + X_2 \frac{d^2}{dt^2}w(t) \\= B_0 r(t) + v(t),
\end{multline}
where $r(t)$ is one-dimensional and known, and $B_0$ has dimension $4\times1$ and has only the first element nonzero. Figure \ref{fig:sim} shows the structure of this network, where it can be seen that the excitation signal $r(t)=r_1(t)$ enters the network only at node $w_1$. One can think of this network as a mechanical mass-spring-damper network as explained in Section \ref{sec:phys}, where $X_0$ and $Y_0$ contain the spring constants, $X_1$ and $Y_1$ contain the damper coefficients, $X_2$ contains the masses, the node signals $w(t)$ represent the positions of the masses, and the excitation signal $r(t)$ is a force. One can also think of this network as an electrical circuit with nodes that are interconnected through capacitors, resistors, and inductors (in parallel). The matrices $X_0$ and $Y_0$ contain the capacitances, $X_1$ and $Y_1$ contain the conductance values of the resistors, $X_2$ contain the inverses of the inductances, the node signals $w(t)$ represent the electric potentials of the interconnection points, and the excitation signal $r(t)$ is the derivative of a current flow.

The discrete-time representation is obtained by applying the discretization method described in Section \ref{ssec:dt} with sampling frequency $f_s = 100$ Hz. In addition, the disturbance $v(t)$ acting on the network is modeled in discrete time as a white noise filtered by a first order filter. This results in the discrete-time network model \eqref{eq:ABF_DT}
\begin{equation}\label{eq:simdt}
    [A_0 + A_1\q + A_2q^{-2}] w(t) = B_0 r(t) + [I + C_1\q] e(t),
\end{equation}
with $A_i=\bar{X}_i+\bar{Y}_i$ and $A_i, \bar Y_i$ symmetric and $\bar X_i$ diagonal for $i=0,1,2$. The network topology is assumed to be unknown reflected by the situation that in the model there are parametrized second order connections between all pairs of nodes. As $A_2=\bar{X}_2$ is diagonal, identifiability Condition 2 in Proposition \ref{prop:idf} is satisfied. The location where $r(t)$ enters and the first nonzero parameter of $B_0$ is assumed to be known, which induces that $B_0$ is fixed and not parametrized. This guarantees that identifiability Condition 4 in Proposition \ref{prop:idf} is satisfied.

The symmetric structure of $Y(p)$ is taken into account in the parametrization of the continuous-time model. The continuous-time model matrices \eqref{eq:simct} are parametrized as
\begin{equation}
    \resizebox{0.89\hsize}{!}{$X_0 = \begin{bmatrix}\theta^c_{1}&0&0&0\\ 0&\theta^c_{4}&0&0\\ 0&0&\theta^c_{7}&0\\ 0&0&0&\theta^c_{10}\end{bmatrix},\
    X_1 = \begin{bmatrix}\theta^c_{2}&0&0&0\\ 0&\theta^c_{5}&0&0\\ 0&0&\theta^c_{8}&0\\ 0&0&0&\theta^c_{11}\end{bmatrix}, $}
\end{equation}\begin{equation}
    \resizebox{0.45\hsize}{!}{$X_2 = \begin{bmatrix}\theta^c_{3}&0&0&0\\ 0&\theta^c_{6}&0&0\\ 0&0&\theta^c_{9}&0\\ 0&0&0&\theta^c_{12} \end{bmatrix},$}
\end{equation}\begin{equation}
    \resizebox{0.89\hsize}{!}{$Y_0 = \begin{bmatrix}\star&\theta^c_{13}&\theta^c_{15}&\theta^c_{17}\\ \theta^c_{13}&\star&\theta^c_{19}&\theta^c_{21}\\ \theta^c_{15}&\theta^c_{19}&\star&\theta^c_{23}\\ \theta^c_{17}&\theta^c_{21}&\theta^c_{23}&\star\end{bmatrix},\
    Y_1 = \begin{bmatrix} \star&\theta^c_{14}&\theta^c_{16}&\theta^c_{18}\\ \theta^c_{14}&\star&\theta^c_{20}&\theta^c_{22}\\ \theta^c_{16}&\theta^c_{20}&\star&\theta^c_{24}\\ \theta^c_{18}&\theta^c_{22}&\theta^c_{24}&\star\end{bmatrix}, $}
\end{equation}
where the elements $\star$ follow from the Laplacian structure. Observe that $\theta^c_i\geq0$ for $i=1,\ldots,12$ and $\theta^c_i\leq0$ for $i=13,\ldots,24$ as all components have positive values, see Section \ref{ssec:high}. The exact true parameter values, represented by $\theta^{c0}$ are given in Table \ref{tab:parct}.

The external excitation signal $r_1(t)$ is an independent white noise process with mean $0$ and variance $\sigma_r^2=1$. All nodes are subject to disturbances $e_i(t)$, which are independent white noise processes (uncorrelated with $r_1(t)$) with mean $0$ and variance $\sigma_e^2=10^{-4}$. In Step~2 of the algorithm, the possibility to apply the weighting $W(\zn)=P^{-1}(\zn)$ is exploited. In Step~3 of the algorithm, at most $50$ iterations are allowed to improve the result of Step~2.

Experiment~1 serves to show that the parameters can consistently be identified with a single excitation signal only. In order to do so, the identification is performed for different orders $n$ of the ARX model in Step~1 and different data lengths $N$, such that they increase at an appropriate rate, guaranteeing that $n^4/N$ decreases for increasing $n$ and $N$ \cite{GALRINHO2019}. The chosen values $n$, $N$ and the rate $n^4/N$ are given in Table \ref{tab:exp1}. For each experimental set, $20$ Monte-Carlo simulations are performed, where in each run new excitation and noise signals are generated.

Experiment~2 serves to identify the parameters and topology with a single excitation signal only and to show that using more excitation signals at different nodes improves the results. In order to show this, two sets of experiments are performed, one set with a single excitation signal ($K=1$) entering at node $w_1$ and one set with three excitation signals ($K=3$) entering the network at node $w_1$, $w_2$, and $w_3$. In the former case, $B_0$ is a $4\times1$ unit vector that is fully known and in the latter case, $B_0$ is a $4\times3$ selection matrix with $b_{11,0}=1$; with parametrized elements $b_{22,0}=\theta^c_{24}$ and $b_{33,0}=\theta^c_{25}$ with true values $\theta^{c0}_{24}=1$ and $\theta^{c0}_{25}=1$; and with all other elements equal to $0$. The order of the ARX model in Step 1 of the algorithm is $n=5$ and the number of samples generated for each data set is $N = 10,000$. Both experimental sets consists of $100$ Monte-Carlo simulations, where in each run new excitation and noise signals are generated.

\subsection{Simulation results}
The simulation results of Experiment~1 are shown in Figure \ref{fig:rmse}. This figure shows a Boxplot of the relative mean squared error (RMSE) of the continuous-time model parameters, where the RMSE is determined as
\begin{equation}\label{eq:rmse}
    \text{RMSE} = \frac{\|\theta^{c0}-\hat\theta^c_N\|_2^2}{\|\theta^{c0}\|_2^2},
\end{equation}
where $\theta^c$ contains the parameters of $X(p)$ and $Y(p)$. From Figure \ref{fig:rmse} it can be seen that the RMSE decreases if both $n$ and $N$ increase such that the rate $n^4/N$ decreases. This observation supports the statement that consistent identification is achieved if the order of the ARX model $n$ tends to infinity as function of the data length $N$ at an appropriate rate \cite{GALRINHO2019}.

The simulation results of Experiment~2 are shown in Figures \ref{fig:rel1}-\ref{fig:mean3}, and Table \ref{tab:parct}.

Figure \ref{fig:rel1} and \ref{fig:rel3} show a Boxplot of the relative parameter estimation errors for $K=1$ and $K=3$, respectively, for parameters for which their underlying true value is unequal to $0$. For the other parameters, the mean values are provided in Figures \ref{fig:mean1} and \ref{fig:mean3}. From Figures \ref{fig:rel1} and \ref{fig:rel3}, it can be seen that the median of the relative errors is around $0$ for all parameters, which means that the median of the estimated parameters values are close to the true values. This supports the statement that the parameters can be identified with a single excitation signal only. However, Figure \ref{fig:rel1} also shows that for the experiment with a single excitation signal ($K=1$), 50\% of the relative parameter errors are within a range of 10\% deviation. This is quite a large deviation. From Figure \ref{fig:rel3}, it can be seen that this range reduces to 2\% deviation if the number of excitation signals is increased to three ($K=3$). Increasing the number of excitation signals improves the signal-to-noise ratio which has a clear effect on the variance of the estimated parameters.

Table \ref{tab:parct} contains the mean and standard deviation of the estimated model parameters. The experiment with three external excitation signals has two additional parameters $\theta^c_{25}$ and $\theta^c_{26}$, which have true values $1$ and which are estimated with mean $1.000$ and $9.9983\times10^{-1}$, respectively, and standard deviation $6.6505\times10^{-4}$ and $9.2690\times10^{-4}$, respectively. Although the estimates are quite accurate, small biases can still occur because of the finite values of $n$ and $N$. For all parameters, this bias is within a bound of 1 standard deviation.

Figure \ref{fig:mean1} and Figure \ref{fig:mean3} show the true values and the mean estimated values of the parameters, focusing on those parameters whose true value are equal to $0$. For $K=1$ it would be hard to identify the correct topology of the network, i.e. estimate which parameters are unequal to $0$, on the basis of the estimated mean values only. Note that for example, the zero parameter $\theta_{19}^c$ has a mean value that is higher than the non-zero parameter $\theta_3^c$. For $K=3$ this situation improves drastically.

\section{Discussion} \label{sec:disc}
In this section, three extensions of the presented theory are discussed. First, the connection with dynamic networks is made. Second, networks with unmeasured nodes are considered. Third, parameter constraints are discussed.

\subsection{Dynamic networks}
A commonly used description of dynamic networks is the module representation \cite{VANDENHOF2013}, in which a network is considered to be the interconnection of directed transfer functions (modules) through measured node signals as
\begin{equation}\label{eq:modrep}
    w(t) = G(q) w(t) + R(q) r(t) + H(q) \tilde{e}(t),
\end{equation}
with white noise process $\tilde{e}(t)$ and with proper rational transfer function matrices $G(q)$, $R(q)$, and $H(q)\in\mathcal{H}$, where the matrix entries $G_{jk}(q)$, $R_{jk}(q)$, and $H_{jk}(q)$ describe the dynamics in the paths from $w_k(t)$, $r_k(t)$, and $\tilde{e}_k(t)$ to $w_j(t)$, respectively. A diffusively coupled network \eqref{eq:ABF_DT} can be described as a module representation with the following particular symmetrical properties: \cite{KIVITS2019}
\begin{itemize}
    \item The transfer functions $G_{jk}(q)$ and $G_{kj}(q)$ have the same numerator for all $j, k$.
    \item The transfer functions $G_{jk}(q)$ and $R_{jm}(q)$ have the same denominator for all $k,m$.
    \item The transfer functions $G_{jk}(q)$ and $H_{jm}(q)$ have the same denominator for all $k,m$ if $F(q)$ is polynomial.
\end{itemize}
Moreover, conditions for a unique mapping between a module representation and a diffusively coupled network are formulated in \cite{KIVITS2019}.

The structure of $G(q)$ and $R(q)$ corresponding to a diffusively coupled network with three nodes is illustrated by Figure \ref{fig:modrep}. It shows that the modules $G_{jk}(q)=-a_{jj}^{-1}(\q)a_{jk}(\q)$ and $G_{kj}(q)=-a_{kk}^{-1}(\q)a_{jk}(\q)$ have the same numerator ($a_{jk}(\q)$) and all transfer functions in the paths towards a specific node $w_j$ have the same denominator ($a_{jj}(\q)$). Since $G_{jk}(q)$ and $G_{kj}(q)$ have the same numerator, they will either be both present or both absent, which is in accordance with the fact that they represent a single diffusively coupled interconnection. In addition, the connections to the ground node are only present in the denominators, because they are only present in $a_{jj}(\q)$. This means that they do not have an effect on the topology in the module representation, although they are part of the topology in the diffusively coupled network.

\begin{figure}
  \centering
  \includegraphics[page=2,width=0.95\columnwidth]{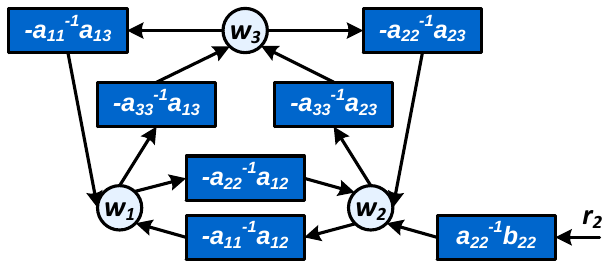}
  \caption{A module representation of a diffusively coupled network with three nodes. } \label{fig:modrep}
  \vspace{-1em}
\end{figure}

\subsection{Partial measurements}
Throughout this paper we assumed that all node signals are measured, which is a situation to which the identification method that we presented is particularly tuned. From literature it is known \cite{FRISWELL1999,LOPES2015} that network identifiability can be achieved if all nodes are measured and one node is excited, or all nodes are excited and one node is measured. The former situation is covered in this paper, see Proposition \ref{prop:idf}, as the latter situation seems less common and leading to higher experimental costs. In this latter situation a different identification method would be required, further exploiting the role of the different excitation signals, leading to a so-called indirect method of identification.

For the situation that only a subset of nodes is measured and/or excited, general identifiability conditions are not yet known, but some particular situations are considered in \cite{BAZANELLA2019} for the case of directed networks.

For estimating only a particular component or a particular connection in the network, the identifiability conditions will be less severe. In the partial measurement situation, unmeasured node signals can be removed from the representation by Gaussian elimination, which is equivalent to Kron reduction \cite{DORFLER2013} and immersion \cite{DANKERS2016}, which has been effectively applied in the local module identification problem of directed networks, see e.g. \cite{DANKERS2016}. Related results for the situation of undirected networks will be reported in a follow-up paper.

\subsection{Parameter constraints}
Physical networks that consist of interconnected physical components, such as mass-spring-damper systems and RLC circuits, are known to have positive real-valued component values in the continuous-time representation \eqref{eq:xy}. This also leads to coefficients with known signs in the corresponding discrete-time representation. In both the consistency proof and the presented algorithm, these sign constraints are not taken into account. They can be taken into account in the algorithm by adding inequality constraints of the form $\Gamma_u\vartheta<0$ to the optimization problems \eqref{eq:opt1}-\eqref{eq:constr1} and \eqref{eq:opt2}-\eqref{eq:constr2}. A priori known parameter values can easily be taken into account by the equality constraint $\Gamma\vartheta=\gamma$ in the optimization problems \eqref{eq:opt1}-\eqref{eq:constr1} and \eqref{eq:opt2}-\eqref{eq:constr2}. Known (continuous-time) component values can be taken into account as well, by splitting $\vartheta$ as $\vartheta=\vartheta_u+\vartheta_k$, where $\vartheta_u$ and $\vartheta_k$ represent the unknown and known part of $\vartheta$, respectively. Then the linear form $-Q(\zeta^n)\vartheta=0$ leads to $-Q(\zeta^n)\vartheta_u+\gamma_k(\zeta^n,\vartheta_k)=0$, where $\gamma_k(\zeta^n,\vartheta_k):=-Q(\zeta^n)\vartheta_k$ is known.

\section{Conclusion} \label{sec:conc}
Undirected networks of diffusively coupled systems can be represented by polynomial representations with particular structural properties. This has enabled the development of an effective prediction error identification method for identifying the physical components and the topology of the network. Conditions for consistent parameter estimates have been formulated for the situation that all network nodes are measured, showing that only a single excitation signal is needed for consistency. The identification is performed through a multi-step algorithm that relies on convex optimizations, being a reworked version of the recently introduced weighted null-space fitting method, adapted to the situation of the structured network models. The results of identifying the topology and parameters of a network are illustrated in a Monte Carlo simulation example. It shows that, while consistency is guaranteed for a single excitation signal, the variance of parameter estimates improves considerably when increasing the number of excitations.


\appendix
\subsection{Parameters of the structured network} \label{app:struct}
Remember that a polynomial matrix $A(\q)$ as its $(i,j)$-th element has $a_{ij}(\q)=\sum_{\ell=0}^{n_a} a_{jk,\ell}q^{-\ell}$. The model structure $A(\q,\theta_a)$, $B(\q,\theta_b)$, $C(\q,\theta_c)$ of the network model \eqref{eq:armax} is parametrized in terms of the parameters $\theta_a$, $\theta_b$, and $\eta_c$, where $\bar{C}(\q,\eta_c,\theta_a) = C(\q,\theta_c) A_0(\theta_a)$, having its constant term parametrized by $\theta_a$ and its dynamic terms by parameters $\theta_c$. $A(\q)$ is parametrized symmetrically. The parameter vectors $\theta_a$, $\theta_b$, $\theta_c$, and $\eta_c$ are given by
\begin{alignat}{3}
  &\theta_a = \begin{bmatrix}\theta_{a_1}\\ \theta_{a_2}\\ \vdots\\ \theta_{a_L}\end{bmatrix},\quad
  &\theta_{a_i}&=\begin{bmatrix}\theta_{a_{i\bm{i}}}\\ \theta_{a_{i\bm{(i+1)}}}\\ \vdots\\ \theta_{a_{iL}}\end{bmatrix},\quad
  &\theta_{a_{ij}}&=\begin{bmatrix}a_{ij,0}\\ a_{ij,1}\\ \vdots\\ a_{ij,n_a}\end{bmatrix}, \label{eq:ThetaA}\\
  &\theta_b = \begin{bmatrix}\theta_{b_1}\\ \theta_{b_2}\\ \vdots\\ \theta_{b_L}\end{bmatrix},\quad
  &\theta_{b_i}&=\begin{bmatrix}\theta_{b_{i1}}\\ \theta_{b_{i2}}\\ \vdots\\ \theta_{b_{iK}}\end{bmatrix},\quad
  &\theta_{b_{ij}}&=\begin{bmatrix}b_{ij,0}\\ b_{ij,1}\\ \vdots\\ b_{ij,n_b}\end{bmatrix}, \label{eq:ThetaB}\\
  &\theta_c = \begin{bmatrix}\theta_{c_1}\\ \theta_{c_2}\\ \vdots\\ \theta_{c_L}\end{bmatrix},\quad
  &\theta_{c_i}&=\begin{bmatrix}\theta_{c_{i1}}\\ \theta_{c_{i2}}\\ \vdots\\ \theta_{c_{iL}}\end{bmatrix},\quad
  &\theta_{c_{ij}}&=\begin{bmatrix}c_{ij,1}\\ c_{ij,2}\\ \vdots\\ c_{ij,n_c}\end{bmatrix},\label{eq:ThetaC}\\
  &\eta_c = \begin{bmatrix}\eta_{c_1}\\ \eta_{c_2}\\ \vdots\\ \eta_{c_L}\end{bmatrix},\quad
  &\eta_{c_i}&=\begin{bmatrix}\eta_{c_{i1}}\\ \eta_{c_{i2}}\\ \vdots\\ \eta_{c_{iL}}\end{bmatrix},\quad
  &\eta_{c_{ij}}&=\begin{bmatrix}\bar{c}_{ij,1}\\ \bar{c}_{ij,2}\\ \vdots\\ \bar{c}_{ij,n_c}\end{bmatrix}. \label{eq:EtaC}
\end{alignat}

\subsection{ARX parametrization and regressor} \label{app:arx}
The model structure $\breve{A}(\q,\zeta_a^n)$ and $\breve{B}(\q,\zeta_b^n)$ of the nonparametric ARX model \eqref{eq:arxInf} is parametrized in terms of the parameters $\zeta^n$. The parameter vector $\zeta^n:=\begin{bmatrix}[\zeta_a^n]^{\top}&[\zeta_b^n]^{\top}\end{bmatrix}^{\top}$ is given by
\begin{alignat}{3}
  &\zeta^n_a = \begin{bmatrix}\zeta^n_{a_1}\\ \zeta^n_{a_2}\\ \vdots\\ \zeta^n_{a_L}\end{bmatrix},\quad
  &\zeta^n_{a_i}&=\begin{bmatrix}\zeta^n_{a_{i1}}\\ \zeta^n_{a_{i2}}\\ \vdots\\ \zeta^n_{a_{iL}}\end{bmatrix},\quad
  &\zeta^n_{a_{ij}}&=\begin{bmatrix}\breve{a}_{ij,1}\\ \breve{a}_{ij,2}\\ \vdots\\ \breve{a}_{ij,n}\end{bmatrix},\label{eq:arxZetaA}\\
  &\zeta^n_b = \begin{bmatrix}\zeta^n_{b_1}\\ \zeta^n_{b_2}\\ \vdots\\ \zeta^n_{b_L}\end{bmatrix},\quad
  &\zeta^n_{b_i}&=\begin{bmatrix}\zeta^n_{b_{i1}}\\ \zeta^n_{b_{i2}}\\ \vdots\\ \zeta^n_{b_{iK}}\end{bmatrix},\quad
  &\zeta^n_{b_{ij}}&=\begin{bmatrix}\breve{b}_{ij,0}\\ \breve{b}_{ij,1}\\ \vdots\\ \breve{b}_{ij,n}\end{bmatrix}.\label{eq:arxZetaB}
\end{alignat}
The regressor $[\varphi^n(t)]^{\top}$ in \eqref{eq:peARX2} is given by $[\varphi^n(t)]^{\top} = \begin{bmatrix}-[\varphi^n_w(t)]^{\top}& [\varphi^n_r(t)]^{\top}\end{bmatrix}$ with
\begin{align}
  [\varphi^n_w(t)]^{\top} &= \begin{bmatrix}[\varphi^n_{w_1}(t)]^{\top} \!&\! [\varphi^n_{w_2}(t)]^{\top} \!&\!\cdots\!&\! [\varphi^n_{w_L}(t)]^{\top}\end{bmatrix}, \\
  [\varphi^n_{w_i}(t)]^{\top} &= \begin{bmatrix}w_i(t-1) \!&\! w_i(t-2) \!&\! \cdots \!&\! w_i(t-n)\end{bmatrix}, \label{eq:arxRegW}\\
  [\varphi^n_r(t)]^{\top} &= \begin{bmatrix}[\varphi^n_{r_1}(t)]^{\top} \!&\! [\varphi^n_{r_2}(t)]^{\top} \!&\! \cdots \!&\! [\varphi^n_{r_K}(t)]^{\top}\end{bmatrix}, \\
  [\varphi^n_{r_i}(t)]^{\top} &= \begin{bmatrix}r_i(t) \!&\! r_i(t\! -\! 1) \!&\! r_i(t\! - \!2) \!&\! \cdots \!&\! r_i(t\! -\! n)\end{bmatrix}. \label{eq:arxRegR}
\end{align}

\subsection{Matrices $Q(\zo)$ and $T(\vartheta^0)$} \label{app:QT}
In order to construct $Q(\zo)$ and $T(\vartheta^0)$, we first define some other matrices.\\

\subsubsection{Zero and identity}
Let $0_{i,j}$ denote a matrix of dimension $i\times j$ with all its elements equal to $0$. Let $I_{i,j}$ denote an identity matrix of dimension $i\times j$, where $I_{i,j}=\begin{bmatrix}I_{i,i}&0_{i,j-i}\end{bmatrix}$ for $i\leq j$ and $I_{i,j}=\begin{bmatrix}I_{j,j}&0_{j,i-j}\end{bmatrix}^{\top}$ for $i\geq j$. Let $I_{k(i,j)}$ denote a block diagonal matrix of $k$ blocks of $I_{i,j}$ and let $I_{\ell(k(i,j))}$ denote a block diagonal matrix of $\ell$ blocks of $I_{k(i,j)}$.

\subsubsection{$\Pi$} Define the matrices
\begin{equation} \label{eq:PiAB}
  \Pi_i^a := \begin{bmatrix}\zeta_{a_{i1}^n}&~0_{n,n_a}\\ \vdots&\vdots\\ \zeta_{a_{i(i-1)}^n}&~0_{n,n_a}\\ \zeta_{a_{ii}^n}&-I_{n,n_a}\\ \zeta_{a_{i(i+1)}^n}&~0_{n,n_a}\\ \vdots&\vdots\\ \zeta_{a_{iL}^n}&~0_{n,n_a} \end{bmatrix},\
  \Pi_i^b := \begin{bmatrix}\zeta_{b_{i1}^n}&0_{n+1,n_a}\\ \zeta_{b_{i2}^n}&0_{n+1,n_a}\\ \vdots&\vdots\\ \zeta_{b_{iK}^n}&0_{n+1,n_a} \end{bmatrix},
\end{equation}
and observe that $\Pi_i^a$ has dimensions $Ln\times (n_a+1)$ and that $\Pi_i^b$ has dimensions $K(n+1)\times (n_a+1)$.

For $x\in\{a,b\}$, define the block matrix
\begin{equation}
  \bar{\Pi}_L^x := \begin{bmatrix}
  Z^x_{0,L}           &Z^x_{1,L-1}         &\cdots  &Z^x_{L-1,1}\\
  R(\Pi_1^x,\Pi_L^x)  &R(\Pi_2^x,\Pi_L^x)  &\cdots  &R(\Pi_L^x,\Pi_L^x)\\
  S_{L-1}(\Pi_1^x)    &S_{L-2}(\Pi_2^x)    &\cdots  &S_{L-L}(\Pi_L^x)\end{bmatrix},
\end{equation}
with $Z^a_{i,j}$ an $i\times j$ block matrix with blocks $0_{Ln,(n_a+1)}$ and with $Z^b_{i,j}$ an $i\times j$ block matrix with blocks $0_{K(n+1),(n_a+1)}$ (that is $Z^a_{i,j} := 0_{iLn,j(n_a+1)}$ and $Z^b_{i,j} := 0_{iK(n+1),j(n_a+1)}$), with \begin{equation}
  R(\Pi_i^x,\Pi_j^x) := \begin{bmatrix}\Pi_i^x&\Pi_{i+1}^x&\cdots&\Pi_j^x\end{bmatrix},
  \quad\text{for } i\leq j,
\end{equation}
and with
\begin{equation}
  S_{i}(\Pi_j^x) := \begin{bmatrix}Z^x_{i,1}&D_{i}(\Pi_j^x)\end{bmatrix},
\end{equation}
with $D_{i}(\Pi_j^x)$ a block diagonal matrix consisting of $i$ blocks of $\Pi_j^x$. Observe that $\bar{\Pi}_L^a$ has dimensions $L^2n\times \frac{1}{2}L(L+1)(n_a+1)$ and $\bar{\Pi}_L^b$ has dimensions $LK(n+1)\times \frac{1}{2}L(L+1)(n_a+1)$.

\subsubsection{Toeplitz}
Let $\mathcal{T}_{i,j}(x)$ denote a Toeplitz matrix of dimension $i\times j$ (with $i\geq j$) given by
\begin{equation}\label{eq:Toeplitz}
  \mathcal{T}_{i,j}(x) :=
  \mathcal{T}_{i,j}(\begin{bmatrix}x_0&\!\!x_1&\!\!\!\cdots&\!\!\!x_{i-1}\end{bmatrix}) = \begin{bmatrix}x_0&\!\!~&\!\!0\\ \vdots&\!\!\ddots&\!\!~\\ x_{j-1}&\!\!\cdots&\!\!x_0\\ \vdots&\!\!\ddots&\!\!\vdots\\ x_{i-1}&\!\!\cdots&\!\!x_{i-j-1} \end{bmatrix}.
\end{equation}

Define the following Toeplitz matrices of dimension $k\times\ell$:
\begin{align}
  \mathcal{T}_{k,\ell}(\breve{a}_{ij}) &:= \mathcal{T}_{k,\ell}(\begin{bmatrix}\breve{a}_{ij,0}&\breve{a}_{ij,1}&\cdots&\breve{a}_{ij,k} \end{bmatrix}), \\
  \mathcal{T}_{k,\ell}(\breve{b}_{ij}) &:= \mathcal{T}_{k,\ell}(\begin{bmatrix}0&\breve{b}_{ij,0}&\breve{b}_{ij,1}&\cdots&\breve{b}_{ij,k} \end{bmatrix}), \\
  \mathcal{T}_{k,\ell}(\bar{c}_{ij}) &:= \mathcal{T}_{k,\ell}(\begin{bmatrix}a_{ij,0}&\bar{c}_{ij,1}&\cdots&\bar{c}_{ij,k}\end{bmatrix}),
\end{align}
Note that for $\breve{a}_{ij}$ it is known that $\breve{a}_{ij,0}=1$ for $i=j$ and $\breve{a}_{ij,0}=0$ for $i\neq j$; and note that for $\bar{c}_{ij}$ it is known that $\bar{c}_{ij,0}=a_{ij,0}=a_{ji,0}$ and $\bar{c}_{ij,k}=0$ for $k>n_c$.

Define the following block matrices
\begin{align}
  \bar{\mathcal{T}}_{k,\ell}(\breve{A}) &:= \begin{bmatrix} \mathcal{T}_{k,\ell}(\breve{a}_{11})&\cdots&\mathcal{T}_{k,\ell}(\breve{a}_{L1})\\ \vdots&~&\vdots\\ \mathcal{T}_{k,\ell}(\breve{a}_{1L})&\cdots&\mathcal{T}_{k,\ell}(\breve{a}_{LL}) \end{bmatrix}, \\
  \bar{\mathcal{T}}_{k,\ell}(\breve{B}) &:= \begin{bmatrix} \mathcal{T}_{k,\ell}(\breve{b}_{11})&\cdots&\mathcal{T}_{k,\ell}(\breve{b}_{L1})\\ \vdots&~&\vdots\\ \mathcal{T}_{k,\ell}(\breve{b}_{1K})&\cdots&\mathcal{T}_{k,\ell}(\breve{b}_{LK}) \end{bmatrix},
\end{align}
where $\bar{\mathcal{T}}_{k,\ell}(\breve{A})$ has dimensions $Lk\times L\ell$ and $\bar{\mathcal{T}}_{k,\ell}(\breve{B})$ has dimensions $Kk\times L\ell$.

For $x\in\{a,b\}$, let $\bar{\mathcal{T}}_{m(k,\ell)}(\breve{X})$ denote a block diagonal matrix consisting of $m$ blocks of $\bar{\mathcal{T}}_{k,\ell}(\breve{X})$. Observe that $\bar{\mathcal{T}}_{n,n_c}(\breve{A})$ has dimension $Ln\times Ln_c$ and that $\bar{\mathcal{T}}_{n+1,n_c}(\breve{B})$ has dimension $K(n+1)\times Ln_c$.

Let $\mathcal{T}_{m(k,\ell)}(\bar{c}_{ij})$ denote a block diagonal matrix consisting of $m$ blocks of $\mathcal{T}_{k,\ell}(\bar{c}_{ij})$. Observe that $\mathcal{T}_{L(n,n)}(\bar{c}_{ij})$ is an $Ln\times Ln$ block diagonal matrix consisting of $L$ blocks of $\mathcal{T}_{n,n}(\bar{c}_{ij})$ and that $\mathcal{T}_{K(n+1,n+1)}(\bar{c}_{ij})$ is an $K(n+1)\times K(n+1)$ block diagonal matrix consisting of $K$ blocks of $\mathcal{T}_{n+1,n+1}(\bar{c}_{ij})$. Finally, define
\begin{equation}
  T_{m(k,\ell)}(\bar{C}) := \begin{bmatrix}
  \mathcal{T}_{m(k,\ell)}(\bar{c}_{11})& \cdots& \mathcal{T}_{m(k,\ell)}(\bar{c}_{1L})\\
  \vdots&~&\vdots\\
  \mathcal{T}_{m(k,\ell)}(\bar{c}_{L1})& \cdots& \mathcal{T}_{m(k,\ell)}(\bar{c}_{LL})
\end{bmatrix}.
\end{equation}

\subsubsection{Matrix $Q(\zo)$} With the matrices defined above, we can now describe the matrix $Q(\zo)$ in \eqref{eq:reg} by
\begin{equation}\label{eq:Q1}
  Q(\zo) = \begin{bmatrix}
    \bar{\Pi}_L^a & 0                  &\bar{\mathcal{T}}_{L(n,n_c)}(\breve{A}^0)\\
    \bar{\Pi}_L^b &-I_{L(K(n+1,n_b))}  &\bar{\mathcal{T}}_{L(n+1,n_c)}(\breve{B}^0)
  \end{bmatrix},
\end{equation}
which has row dimension $L^2n + LK(n+1)$ and column dimension $\frac{1}{2}L(L+1)(n_a+1) + LK(n_b+1) + L^2n_c$.

\subsubsection{Matrix $T(\vartheta^0)$} With the matrices defined above, we can now describe the matrix $T(\vartheta^0)$ in \eqref{eq:reg2} by
\begin{equation}\label{eq:T}
  T(\vartheta^0) = \begin{bmatrix}
    -T_{L(n,n)}(\bar{C}^0) & 0 \\ 0 & -T_{K(n+1,n+1)}(\bar{C}^0) \\
  \end{bmatrix},
\end{equation}
which has dimensions $[L^2n + LK(n+1)]\times [L^2n + LK(n+1)]$.

\bibliographystyle{IEEEtran}
\bibliography{LibraryJournal2021}
\balance

\begin{IEEEbiography}[{\includegraphics[width=1in,height=1.25in,clip,keepaspectratio]{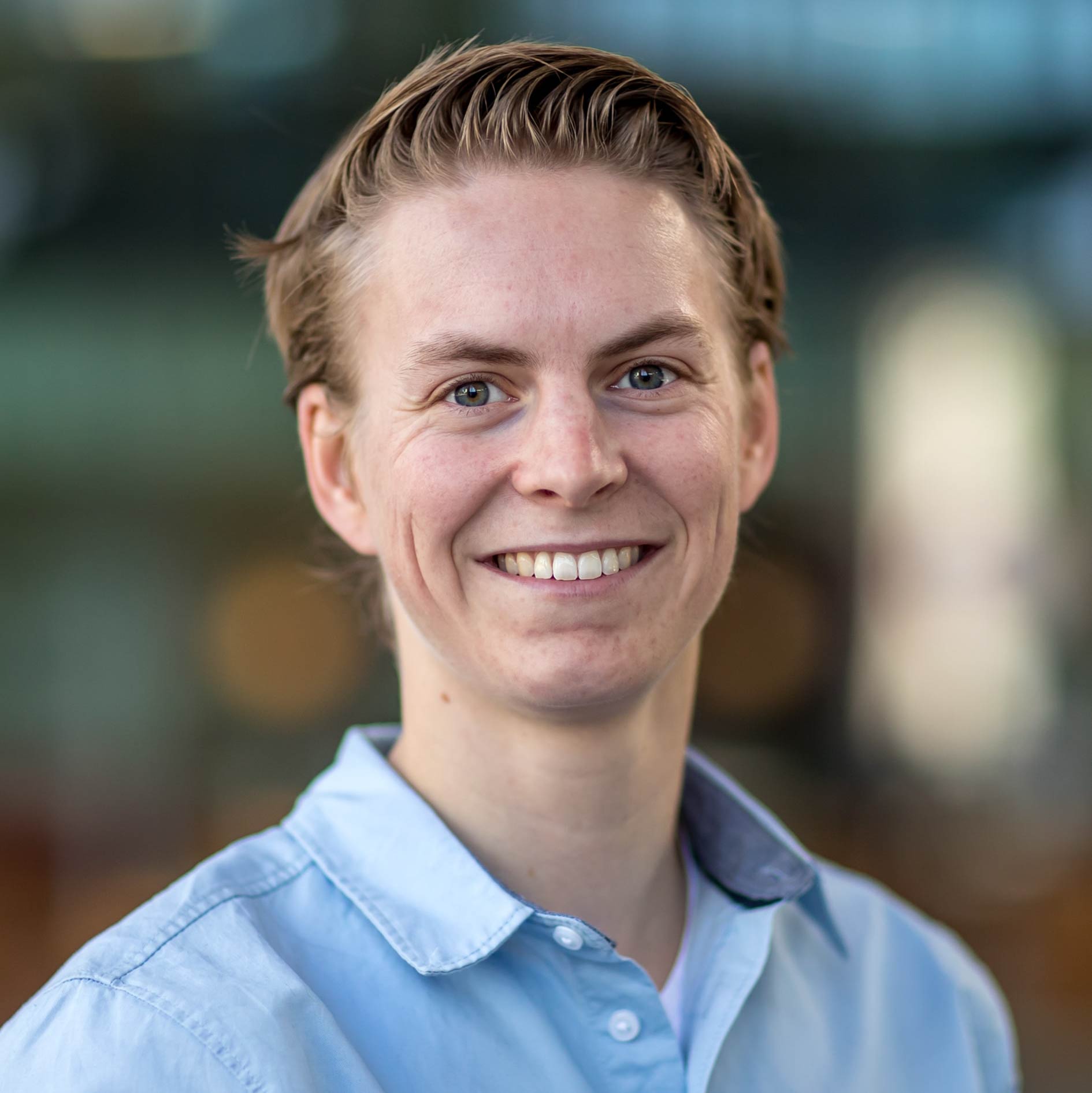}}]{Lizan Kivits} received her bachelor's degree in Electrical Engineering and her master's degree in Systems and Control (with great appreciation) from Eindhoven University of Technology, Eindhoven, The Netherlands, in 2014 and 2017, respectively. In 2017, she started as a researcher in the Control Systems research group at the Department of Electrical Engineering, Eindhoven University of Technology, where she continued her research project as a Ph.D. researcher in 2018. Her research interests include first-principles modeling, data-driven modeling, and dynamic network identification.
\end{IEEEbiography}

\vspace{-2em}

\begin{IEEEbiography}[{\includegraphics[width=1in,height=1.25in,clip,keepaspectratio]{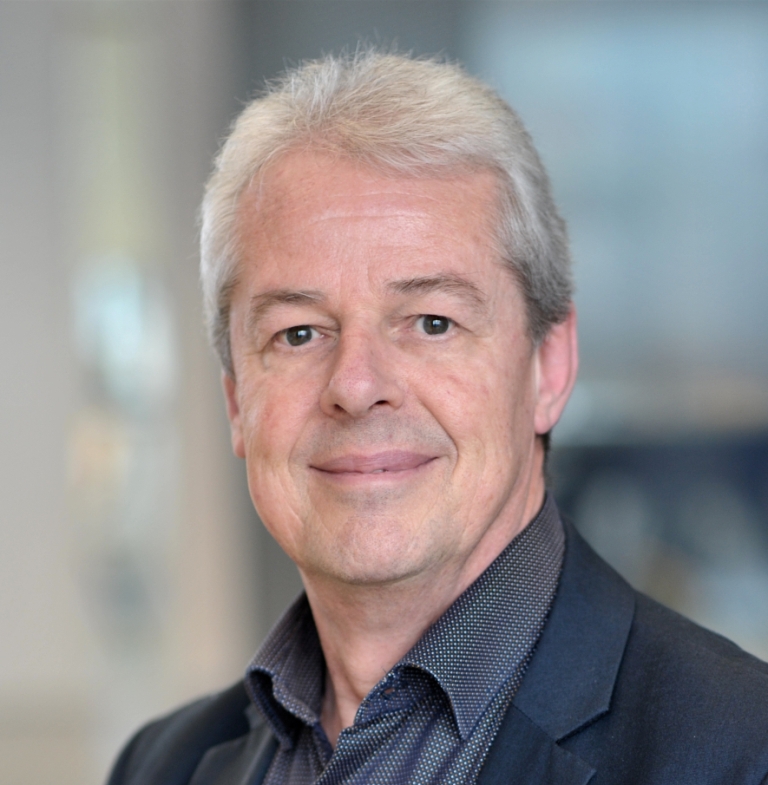}}]{Paul Van den Hof} received the M.Sc. and Ph.D. degrees in electrical engineering from Eindhoven University of Technology, Eindhoven, The Netherlands, in 1982 and 1989, respectively. In 1986 he moved to Delft University of Technology, where he was appointed as Full Professor in 1999. From 2003 to 2011, he was founding co-director of the Delft Center for Systems and Control (DCSC). As of 2011, he is a Full Professor in the Electrical Engineering Department, Eindhoven University of Technology. His research interests include data-driven modeling, identification for control, dynamic network identification, and model-based control and optimization, with applications in industrial process control systems and high-tech systems. He holds an ERC Advanced Research grant for a research project on identification in dynamic networks. Paul Van den Hof is an IFAC Fellow and IEEE Fellow, Honorary Member of the Hungarian Academy of Sciences, and IFAC Advisor. He has been a member of the IFAC Council (1999–2005, 2017-2020), the Board of Governors of IEEE Control Systems Society (2003–2005,2022-2024), and an Associate Editor and Editor of Automatica (1992–2005). In the triennium 2017-2020 he served as Vice-President of IFAC.
\end{IEEEbiography}

\end{document}